\documentclass[a4paper,12pt]{article}
\usepackage{amsmath}
\usepackage{graphicx}
\usepackage{enumitem}
\usepackage{url} 

\graphicspath{{./art/}}

\usepackage{xcolor}
\usepackage[normalem]{ulem}

\usepackage{amsthm}
\usepackage{amsfonts}
\usepackage{mathtools}
\usepackage{appendix}
\usepackage[]{lscape}

\newtheorem{proposition}{Proposition}
\newtheorem{theorem}{Theorem}
\newtheorem{corollary}{Corollary}

\theoremstyle{remark}

\theoremstyle{definition}
\newtheorem{definition}{Definition}
\newtheorem{example}{Example}
\newtheorem{example*}{Example}

\title{Nonparametric regression with dependent censoring or competing risks}

\author{
Jia-Han Shih\footnote{National Sun Yat-sen University, Department of Applied Mathematics, 80424, Kaohsiung, Taiwan, E--mail: jhshih@math.nsysu.edu.tw} \\
Simon M.S. Lo\footnote{United Arab Emirates University, Department of Economics and Finance, E--mail: losimonms@yahoo.com.hk} \\
Ralf A. Wilke\footnote{Copenhagen Business School, Department of Economics, Porcel{\ae}nshaven 16A, 2000 Frederiksberg, DK, E--mail: rw.eco@cbs.dk}
}

\begin{document}

\newcommand{\ds}{\displaystyle}
\newcommand{\nin}{\noindent}
\newcommand{\1}{1\hskip-0.28em \text{I}}
\newcommand{\non}{\nonumber}

\def\spacingset#1{\renewcommand{\baselinestretch}%
{#1}\small\normalsize} \spacingset{1}

\def\T{{ \mathrm{\scriptscriptstyle T} }}
\def\Pr{{ \mathrm{Pr} }}
\def\E{{ \mathrm{E} }}
\def\Var{{ \mathrm{Var} }}
\def\Cov{{ \mathrm{Cov} }}
\def\Corr{{ \mathrm{Corr} }}
\def\BSbeta{{ \boldsymbol{\beta} }}
\def\BSeta{{ \boldsymbol{\eta} }}
\def\BStheta{{ \boldsymbol{\theta} }}
\def\BSphi{{ \boldsymbol{\phi} }}
\def\21{{ 2 \mid 1 }}


\spacingset{1.9} 
\setlength{\textwidth}{16cm} 

\maketitle
\thispagestyle{empty}

\begin{abstract}
Single-index models or time-to-event models are frequently applied in empirical research. These models are non-identifiable in presence of unknown (dependent) censoring or competing risks and do not give informative results in empirical analysis unless rather strong, non-testable restrictions hold. Little is known, whether the known robustness properties of the single-index model carry over to models with dependent censoring or competing risks. This paper shows that the ratio of partial covariate effects on the margins is identifiable in nonparametric models with unknown dependent censoring or nonparametric competing risks models with nonparametric dependence structure, provided an exclusion restriction holds. Commonly used (semi)parametric models for the margin and independent censoring, such as Cox proportional hazards, accelerated failure time or proportional odds models, can be used to obtain relative covariate effects despite their misspecified censoring mechanism. Several nonparametric estimators for the general model are introduced and their numerical properties are studied.\\
\nin Keywords: dependent competing risks, copula, identifiability
\end{abstract}

\section{Introduction}
The non-identifiability of the dependent competing risks model as shown by Cox (1962), Tsiatis (1975) and Wang (2014) is a crux in two respects. First, in the absence of identifying restrictions, bounds on margins and covariate effects are wide and uninformative (Peterson, 1976). Second, any kind of identifying restriction which leads to more informative results is non-testable. Various approaches to identifiability have been considered in the literature. They require restrictions on either the margins, the copula or both. This paper contributes to this literature by showing that the ratio of two partial covariate effects on the margins are identifiable, provided an exclusion restriction holds. This is a very general result as it does not require restrictions on the copula and the margins of the two risks. Many popular single and independent competing risks models are special cases of the model, including the accelerated failure time, proportional odds and Cox proportional hazards models. While Li and Duan (1989) show the robustness of the single-index model with respect to a misspecified link function, this paper shows its robustness with respect to the unknown dependent censoring distribution.

Any identifiability result hinges on a set of restrictions. For example, the copula graphic estimator (Zheng and Klein, 1995, Carri\`{e}re, 1995) requires a fully assumed copula but does not restrict the margins. Hiabu et al. (2025) do not impose restrictions on the margins either. Under valid exclusion restrictions, they show identifiability of the parameter of an Archimedean copula. Lo and Wilke (2024) consider a semiparametric model for a risk of interest, while the other risk or the distribution of dependent censoring is unspecified. They show that the parameter of an Archimedean copula and subsequently the margins are identifiable and estimable. Lo and Wilke (2017) do not impose parametric restrictions on the margins. Under the assumption that the copula does not depend on the covariates, they show that the sign of a covariate effect on the margins is often identifiable. Willems et al. (2025) consider a model with a semiparametric model for the risk of interest, unknown distribution of dependent censoring and unknown censoring. Their approach gives bounds for the parameters which often reveal their sign. In this paper, we introduce an approach that is applicable when the risk of interest, the distribution of dependent censoring, and the copula are completely unspecified. In presence of an exclusion restriction, it is shown that the ratio of two partial covariate effects on the margin of the risk of interest is identifiable. With this result, it is possible to infer the relative size of the effect of a covariate and whether the two partial effects have the same or opposite sign.

Our identifiability result connects to contributions in the econometric literature which have found that some patterns of single risk models can be robustly estimated when frailty is misspecified or omitted in the empirical specification (Lancaster, 1985, Ridder, 1987). Nicoletti and Rondinelli (2010) find this for the ratio of two regression parameters. Modelling the joint frailty distribution corresponds to shaping the copula. Our result connects to these findings as their (semi-)parametric single risk models with independent censoring are special cases of our dependent competing risks model. Our approach is more general, though, as the ratio of partial covariate effects can be identified for any nonparametric model provided the exclusion restriction holds. We verify that the identifiable ratio of partial covariate effects corresponds to the ratio of regression parameters in popular survival models such as the Cox proportional hazards model, the proportional odds model or the piecewise-constant exponential model. By comparing the ratio of partial covariate effects from our nonparametric approach to the ratio of estimated parameters obtained by commonly used semi- or parametric models, it is possible to obtain evidence on whether the latter are correctly specified single-index models in the sense that the true model is a single-index model or that the single-index is correctly specified. This can be tested without knowledge of the link function and the censoring mechanism and without the model being identifiable.

The paper is organised as follows. Section \ref{sec:mod} presents the model and the main identifiability results. Section \ref{sec:examples} illustrates the results with several examples. Section \ref{sec:est} introduces several estimators for the ratio of partial covariate effects. Section \ref{sec:sim} presents numerical results of monte carlo studies which demonstrate the finite-sample behaviour and convergence of the suggested approaches. Section \ref{sec:app} provides empirical analysis with real data to demonstrate the applicability.

\section{The model}\label{sec:mod}
Let $T_1$ and $T_2$ be the marginal continuous survival times of risks 1 and 2, respectively. We assume a so-called exclusion restriction such that the continuous covariates $X$ and $Y$ are included in the risk 1 margin only. Both margins may depend on additional covariates $Z$, but we omit them in the model presentation to simplify the notation. Given the covariates $X = x$ and $Y = y$, by Sklar's Theorem, there always exists a unique copula $C:[0,1]^2 \to [0,1]$ such that
\begin{align}\label{copula}
\Pr ( T_1 > t_1, T_2 > t_2 \mid X = x, Y = y ) = C ( S_1 (t_1 \mid x,y), S_2 (t_2) ),
\end{align}
where
\begin{align*}
S_1 (t_1 \mid x,y) = \Pr ( T_1 > t_1 \mid X = x, Y = y), \quad S_2 (t_2) = \Pr ( T_2 > t_2).
\end{align*}
All $S_1 (t_1 \mid x,y)$, $S_2 (t_2)$, and $C(u,v)$ are nonparametric. Note that the covariates $X$ and $Y$ have no impact on either the dependence structure ($C$) or the risk 2 margin ($S_2$). Exclusion restrictions have already been used in the literature to show identifiability of the competing risks model. The model of Hiabu et al.\ (2025) also allows for nonparametric $S_1$ and $S_2$, but requires that $S_1$ does not depend on $X$ and $S_2$ does not depend on $Y$, and that $C$ is a known Archimedean copula with unknown dependence parameter $\theta$. They show identifiability of $\theta$ or the degree of dependence between risks under these restrictions.

The model defined in (\ref{copula}) has such mild restrictions that it is not identifiable. It is well known that the identifiability of models with nonparametric margins for the two risks requires a known copula including its parameters (Zheng and Klein, 1995, Carri\`{e}re, 1995). The goal of the paper can therefore not be to show identifiability of $S_1$, $S_2$, or $C$. Instead, it focuses on other functions in relation to partial effects which can be of high interest in applications as well. For this purpose, we first introduce the notion of relative covariate effects, which remains meaningful irrespective of any model assumptions (Subsections 2.1 and 2.2). Then, we show that the relative covariate effects have a broader interpretation under the model (\ref{copula}), i.e., the exclusion restriction (Subsection \ref{sec:excl-restr}).

\subsection{Identifiability of relative covariate effects}\label{sec:ident}
Let $T = \min (T_1,T_2)$ be the overall survival time and its survival function is defined by
\begin{align}\label{os}
\pi (t \mid x,y) = \Pr ( T > t \mid X = x, Y = y ).
\end{align}
Differentiating (\ref{os}) with respect to $x$ gives the covariate effect of $x$ on $\pi$
\begin{align*}
\pi'_x (t \mid x,y)
&= \frac{\partial}{\partial x} \pi (t \mid x,y),
\end{align*}
where $\partial_1 C (u,v) = \partial C (u,v) / \partial u$ is the first-order partial derivative of the copula $C$ with respect to the first argument. The first-order partial derivative of any copula exists almost everywhere and no condition in relation to this is added. Similarly, differentiating (\ref{os}) with respect to $y$, we obtain
\begin{align*}
\pi'_y (t \mid x,y)
&= \frac{\partial}{\partial y} \pi (t \mid x,y).
\end{align*}
Note that the partial derivatives $\pi'_x (t \mid x,y)$ and $\pi'_y (t \mid x,y)$ reflect how the covariate $X$ and $Y$ affect the overall survival $T$, as they represent the rates of change of $\pi (t \mid x,y)$ with respect to $X$ and $Y$, respectively. We define a relative covariate effect on $\pi (t \mid x,y)$  as follows.

\begin{definition}\label{d1}
The relative covariate effect based on the partial derivatives of $\pi (t \mid x,y)$  is defined by
\begin{align}\label{proportion}
\eta_\pi( t,x,y ) = \frac{\pi'_x (t \mid x,y)}{\pi'_y (t \mid x,y)}.
\end{align}
\end{definition}

\begin{proposition}\label{g}
The relative covariate effects $\eta_\pi ( t,x,y )$ in (\ref{proportion}) is identifiable.
\end{proposition}
It is clear that Proposition~\ref{g} holds since the overall survival $\pi$ is identifiable, and so are its partial derivatives and the relative covariate effects $\eta_\pi ( t,x,y )$.

Let $\Lambda (t \mid x,y) = - \log \pi (t \mid x,y)$ be the cumulative hazard function of the overall survival. It is straightforward to verify that (\ref{proportion}) is equivalent to
\begin{align*}
\eta_\pi ( t,x,y ) = \frac{\Lambda_x' (t \mid x,y)}{\Lambda_y' (t \mid x,y)},
\end{align*}
where
\begin{align*}
\Lambda_x' (t \mid x,y) = \frac{\partial}{\partial x} \Lambda (t \mid x,y), \quad
\Lambda_y' (t \mid x,y) = \frac{\partial}{\partial y} \Lambda (t \mid x,y)
\end{align*}
are the partial derivatives of $\Lambda$. It is therefore possible to identify $\eta ( t,x,y )$ on the grounds of the partial derivatives of $\pi$ or the partial derivatives of $\Lambda$.

\subsection{Using the conditional expectation of the overall survival \label{sec:condmean}}
In survival or competing risks analysis it is common to work with survival or cumulative hazard functions and the relevant statistical methods are well established. It turns out that the identifiability result, similar to Proposition \ref{g}, can be written in terms of a conditional mean function as well. Writing the problem in terms of mean functions can have practical advantages as it facilitates the use of averaging approaches to improve the stability of solutions.

We define the conditional mean of $T$ by
\begin{align}\label{mean}
m (x,y) &= \E( T \mid X = x, Y = y ).
\end{align}
Assume that the partial derivatives of $m (x,y)$ exist, and denote them by
\begin{align*}
m_x' (x,y) &= \frac{\partial}{\partial x} \E( T \mid X = x, Y = y ), \quad m_y' (x,y) = \frac{\partial}{\partial y} \E( T \mid X = x,Y = y ).
\end{align*}
The partial derivatives $m'_x (x,y)$ and $m'_y (x,y)$ reflect how the covariate $X$ and $Y$ affect the conditional mean of $T$, as they represent the rates of change of $m (x,y)$ with respect to $x$ and $y$, respectively. It is readily seen that the conditional mean $m (x,y)$ is identifiable, and so are its derivatives. In the following, we show that $\eta_\pi (t,x,y)$ can also be identified on the grounds of $m'_x (x,y)$ and $m'_y (x,y)$ provided it satisfies invariance restrictions. Let
\begin{align}
\eta_m (x,y) &= \frac{m'_x (x,y)}{m'_y (x,y)}, \label{mean-proportion} \\
\eta &= \frac{\E ( m'_x (X,Y) )}{\E ( m'_y (X,Y))}, \label{avg-mean-proportion}
\end{align}
be the relative covariate effects based on $m'_x (x,y)$ and $m'_y (x,y)$.

\begin{proposition}\label{eta2}
The relative covariate effects $\eta_m ( x,y )$ in (\ref{mean-proportion}) and $\eta$ in (\ref{avg-mean-proportion}) are identifiable.
\end{proposition}

The following theorem studies the relationships between $\eta_\pi (t,x,y)$, $\eta_m (x,y)$, and $\eta$.


\begin{theorem}\label{equiv}
If $\eta_\pi (t,x,y)$ in (\ref{proportion}) does not depend on $t$, then $\eta_\pi (t,x,y) = \eta_m (x,y)$ for all $t$. Furthermore, if $\eta_m (x,y)$ in (\ref{mean-proportion}) does not depend on $x$ and $y$, then $\eta_m (x,y) = \eta$ for all $x,y$.
\end{theorem}

\begin{proof}
If $\eta_\pi (t,x,y)$ does not depend on $t$, i.e., there exists a function $h: \mathbb{R}^2 \to \mathbb{R}$ such that $\eta_\pi (t,x,y) = h(x,y)$ for all $t$, and using (\ref{proportion}), we have
\begin{align}\label{pf-1}
\pi'_x (t \mid x,y) = h( x,y ) \pi'_y (t \mid x,y)
\end{align}
for all $t$. Recall that since $T$ is a nonnegative random variable, the conditional expectation of $T$ given $X = x$ and $Y = y$ can be expressed as
\begin{align*}
m(x,y) = \int_0^\infty \pi ( t \mid x, y ) dt,
\end{align*}
where $f_T ( t \mid x,y )$ is the conditional density of $T$. Take partial derivatives on both sides of the last equation, we obtain
\begin{align*}
m_x'(x,y) = \int_0^\infty \pi_x' ( t \mid x, y ) dt, \quad m_y'(x,y) = \int_0^\infty \pi_y' ( t \mid x, y ) dt.
\end{align*}
Therefore, integrating $t$ on both sides of (\ref{pf-1}) leads to
\begin{align}\label{pf-2}
m'_x (x,y) = h( x,y ) m'_y (x,y),
\end{align}
which implies
\begin{align*}
\eta_\pi (t,x,y) = h( x,y ) = \frac{m'_x (x,y)}{m'_y (x,y)} = \eta_m(x,y).
\end{align*}

Similarly, if $\eta_m(x,y)$ does not depend on $x$ and $y$, i.e., there exists a constant $c$ such that $\eta_m(x,y) = c$ for all $x,y$. Taking expectation of $(X,Y)$ on both sides of (\ref{pf-2}) leads to
\begin{align*}
\E ( m'_x (X,Y) ) = c \E ( m'_y (X,Y) ),
\end{align*}
which implies
\begin{align*}
\eta_m(x,y) = c = \frac{\E ( m'_x (X,Y) )}{\E ( m'_y (X,Y))} = \eta.
\end{align*}
\end{proof}

We emphasise that although $\eta_\pi (t,x,y)$, $\eta_m (x,y)$, and $\eta$ are closely related, they are identifiable through different routes and are each meaningful on their own, which may lead to different empirical performance in practice. For instance, even if $\eta_m(x,y)$ varies with $x$ and $y$, i.e., $\eta_m (x,y) \not= \eta$ in general, estimating $\eta$ remains meaningful and can still be interpreted.

\subsection{The exclusion restriction \label{sec:excl-restr}}
We show that $\eta_\pi$, $\eta_m$, and $\eta$ in (\ref{proportion}), (\ref{mean-proportion}), and (\ref{avg-mean-proportion}), respectively, have a broader interpretation under the exclusion restriction in (\ref{copula}), i.e., the continuous covariates $X$ and $Y$ are included in the risk 1 margin ($T_1$) only. We show that they are not only informative about the ratio of partial effects on the overall survival, but about the ratio of partial effects on the margin of risk 1 as well.


Under the model (\ref{copula}), the overall survival function in (\ref{os}) becomes $\pi (t \mid x,y) = C ( S_1 (t \mid x,y), S_2 (t) )$. Then, the partial derivatives of $\pi$ are
\begin{align*}
\pi'_x (t \mid x,y) = \partial_1 C ( S_1 (t \mid x,y), S_2 (t) ) S'_{1,x} (t \mid x,y), \\ 
\pi'_y (t \mid x,y) = \partial_1 C ( S_1 (t \mid x,y), S_2 (t) ) S'_{1,y} (t \mid x,y), 
\end{align*}
where $\partial_1 C (u,v) = \partial C (u,v) / \partial u$ is the first-order partial derivative of the copula $C$ with respect to the first argument and
\begin{align*}
S'_{1,x} (t \mid x,y) = \frac{\partial}{\partial x} S_1 (t \mid x,y), \quad S'_{1,y} (t \mid x,y) = \frac{\partial}{\partial y} S_1 (t \mid x,y)
\end{align*}
are the partial derivatives of $S_1$. Combining the above results, we arrive at the conclusion that the relative covariate effect on the overall survival as defined in (\ref{proportion}) is the same as the relative covariate effect on the margin of risk 1. In addition, we can also establish results for the risk 1 margin that are related to those in Theorem~\ref{equiv}. Let
\begin{align*}
m_{1}(x,y) &= \E ( T_1 \mid X = x, Y = y )
\end{align*}
be the conditional mean of $T_1$. Assume that the partial derivatives of $m_1 (x,y)$ exist, and denote them by
\begin{align*}
m_{1,x}'(x,y) &= \frac{\partial}{\partial x} \E ( T_1 \mid X = x, Y = y ), \quad m_{1,y}'(x,y) = \frac{\partial}{\partial y} \E ( T_1 \mid X = x, Y = y ).
\end{align*}
We summarise the results in the following theorem.

\begin{theorem}\label{g-ex}
Under the exclusion restriction on $S_2$ in (\ref{copula}), we have
\begin{align}\label{margin-proportion-1}
\eta_\pi (t,x,y) = \frac{ S'_{1,x} (t \mid x,y)}{ S'_{1,y} (t \mid x,y)}.
\end{align}
If $\eta_\pi (t,x,y)$ in (\ref{proportion}) does not depend on $t$, then
\begin{align*} 
\eta_\pi (t,x,y) = \eta_m (x,y) = \frac{ m_{1,x}'(x,y) }{ m_{1,y}'(x,y)}
\end{align*}
holds for all $t$. In addition, if $\eta_\pi (t,x,y)$ in (\ref{proportion}) does not depend on $t$, $x$, and $y$, then
\begin{align*} 
\eta_\pi (t,x,y) = \eta_m (x,y) = \eta = \frac{ \E ( m_{1,x}'(X,Y) ) }{ \E ( m_{1,y}'(X,Y) ) }
\end{align*}
holds for all $t,x,y$.
\end{theorem}


Since $\eta_\pi (t,x,y)$, $\eta_m (x,y)$, and $\eta$ are identifiable by Proposition~\ref{g} and Proposition~\ref{eta2}, Theorem~\ref{g-ex} implies that it is possible to identify the relative importance of the covariates on marginal survival times of the risk that depends on them. Notably, even under the exclusion restriction, one has
\begin{align*}
\eta_m (x,y) \not= \frac{ m_{1,x}'(x,y) }{ m_{1,y}'(x,y)}, \quad \eta \not= \frac{ \E ( m_{1,x}'(X,Y) ) }{ \E ( m_{1,y}'(X,Y) ) }
\end{align*}
in general. The identities hold when $\eta_\pi (t,x,y) = \eta_m (x,y)$ for all $t$ and $\eta_\pi (t,x,y) = \eta$ for all $t,x,y$ as stated in Theorem~\ref{g-ex}.

The model is very general and nests many well-known models as special cases. These include single risk model with independent censoring. Let risk 2 be the censoring time. Whenever the censoring distribution is covariate invariant, it satisfies the exclusion restriction. Many standard survival models require a covariate invariant independent censoring distribution. It is obvious that these models are nested in the model of this paper. The dependence structure in a competing risks model is an alternative way to write a model with dependent risk specific frailties (Lo et al., 2017; Lo et al., 2025). A copula model is therefore an alternative presentation of a dependent frailty model. Frailty terms become independent when there is no risk dependence. Our identifiability result therefore establishes the theoretical basis that certain relative patterns are identifiable without prior knowledge of frailty structures or the copula in our model.

Our findings can be connected to previous related observations in the econometrics literature by Nicoletti and Rondinelli (2010), Ridder (1987) and Lancaster (1985). The latter two study continuous time single risk models and show that the omission of frailty in their models leads to a scaling bias of all coefficients or just the intercept, respectively. Nicoletti and Rondinelli (2010) consider a discrete time single risk duration model with frailty and independent censoring. They provide numerical evidence with the help of simulations that the ratio of regression parameters $\beta_x$ and $\beta_y$ can be well estimated despite misspecification of the frailty distribution or omission of the frailty term, and despite that both estimated coefficients suffer from sizable estimation bias. It turns out that the ratio of the regression parameters $\beta_x$ and $\beta_y$ corresponds to the ratio of $S'_{1,x} (t \mid x,y)$ and $ S'_{1,y} (t \mid x,y)$ (or $m_{1,x}'(x,y)$ and $m_{1,y}'(x,y)$) in these models.

To establish such connection, we employ the so-called single-index model for the risk 1 margin $S_1$, which assumes that $S_1(t \mid x,y)$ depends solely on a linear combination of $x$ and $y$, namely the single-index $x \beta_x + y \beta_y$. The following theorem gives an equivalent statement of the single-index assumption.


\begin{theorem}\label{single-index}
Let $T_1$ be a continuous random variable and $(\beta_x,\beta_y) \not= (0,0)$ be the regression parameters on the continuous covariates $X$ and $Y$, respectively. Then,
\begin{align}\label{eq-si-1}
S_{1,x}'(t \mid x,y) = \frac{\beta_x}{\beta_y} S_{1,y}'( t \mid x,y)
\end{align}
for all $t,x,y$ if and only if there exists a function $G(t, a)$ such that: (i) for each $a$, $t \mapsto G(t, a)$ is a continuous survival function; (ii) for each $t$, $a \mapsto G(t, a)$ is differentiable; and
\begin{align}\label{eq-si-2}
S_1( t \mid x,y) = G( t, x \beta_x + y \beta_y)
\end{align}
for all $t,x,y$.
\end{theorem}

\begin{proof}
Define two new variables:
\begin{align*}
a = x \beta_x + y \beta_y, \quad b = x \beta_y - y \beta_x.
\end{align*}
The above change of variables is invertible and we can express $(x,y)$ in terms of $(s,t)$ by
\begin{align*}
x = x(a,b) = \frac{a \beta_x + b \beta_y}{\beta_x^2+\beta_y^2}, \quad y = y(a,b) = \frac{a \beta_y - b \beta_x}{\beta_x^2+\beta_y^2}.
\end{align*}
This means that every $(a,b)$ corresponds to exactly one $(x,y)$, and vice versa. For a fixed $t$, define
\begin{align*}
M(t,a,b) = S_1 ( t \mid x(a,b),y(a,b)  ).
\end{align*}
Since $S_1(t \mid x,y)$ is differentiable with respect to $x$ and $y$, and the maps $b \mapsto x(a,b)$ and $b \mapsto y(a,b)$ are smooth for each fixed $a$, by the chain rule,
\begin{align*}
\frac{\partial}{\partial b} M(t,a,b)
&= S_{1,x}'(t \mid x,y) \frac{\partial}{\partial b} x(a,b) + S_{1,y}'(t \mid x,y) \frac{\partial}{\partial b} y(a,b) \\
&= \frac{1}{\beta_x^2+\beta_y^2} ( S_{1,x}'(t \mid x,y) \beta_y - S_{1,y}'(t \mid x,y) \beta_x ).
\end{align*}
If (\ref{eq-si-1}) holds, we arrive that $\partial M(t,a,b) / \partial b = 0$ for all $a,b$, that is, the function $M(t,a,b)$ does not depend on $b$. Consequently, there exists a function $G(t,a)$ such that $a \mapsto G(t,a)$ is differentiable and $M(t,a,b) = G(t,a)$ for all $a,b$. We obtain
\begin{align*}
S_1(t \mid x,y) = G(t,a) = G( t, x \beta_x + y \beta_y ),
\end{align*}
hence $t \mapsto G(t,a)$ is a continuous survival function. The other direction is obvious.
\end{proof}

Theorem~\ref{single-index} reveals that the relationship (\ref{eq-si-1}) is inherent in the single-index model. Together with Proposition~\ref{g} and Theorem~\ref{g-ex}, we obtain the following powerful result, which covers a wide range of models in the literature.

\begin{corollary}\label{cor-si}
Under the exclusion restriction, if the risk 1 margin $S_1$ follows the single-index model in (\ref{eq-si-2}), the ratio of regression parameters $\beta_x$ and $\beta_y$ is identifiable, provided that $(\beta_x,\beta_y) \not= (0,0)$. In addition,
\begin{align*}
\frac{\beta_x}{\beta_y} = \eta_\pi(t,x,y) = \eta_m(x,y) = \eta
\end{align*}
for all $t,x,y$, where $\eta_\pi(t,x,y)$, $\eta_m(x,y)$, and $\eta$ are defined in (\ref{proportion})--(\ref{avg-mean-proportion}), respectively.
\end{corollary}

Corollary~\ref{cor-si} shows that the identity $\beta_x / \beta_y = \eta_\pi(t,x,y) = \eta_m(x,y) = \eta$ holds under two conditions: (i) the risk 1 margin $S_1$ follows a single-index model with continuous covariates $X$ and $Y$ and a proper link function $G$; and (ii) the risk 2 margin $S_2$ is covariate-invariant. Importantly, this result holds regardless of other conditions, such as the specific form of $G$ or the dependence structure between the two risks. The conditions (i) and (ii) are rather mild and are satisfied by many well-known statistical regression models. Corollary~\ref{cor-si} therefore implies that $\eta$ is invariant in these models and can be easily obtained by dividing the parameters on the covariates without knowing the link function and the copula. It is already known that the parameters of single-index models with unknown link function are identifiable up to a multiplicative scalar (Li and Duan, 1989). Our result generalises this to situations with dependent censoring or dependent competing risks.


\section{Examples \label{sec:examples}}

We illustrate Corollary~\ref{cor-si} using prominent examples satisfying $\beta_x / \beta_y = \eta$, where $S_1$ follows popular and well-known regression models, $S_2$ is an unknown but covariate-invariant censoring distribution, and $C$ is unknown as well. An example of $\beta_x / \beta_y \not= \eta$ is also provided.


\subsection{Single-index models}

\begin{example}[Cox PH]\label{ex-1}
The Cox proportional hazards (PH) model is given by
\begin{align*} 
S_1 (t \mid x,y) = S_{10}(t)^{\exp ( x \beta_x + y \beta_y  )},
\end{align*}
where $S_{10}(t) = S_1 (t \mid 0,0)$ is the baseline survival function of the risk 1 margin.


\end{example}


\begin{example}[PO]\label{ex-2}
The proportional odds (PO) model is given by
\begin{align*} 
S_1 (t_1 \mid x,y) = \frac{S_{10} (t)}{S_{10} (t) + (1 - S_{10}(t)) \exp ( x \beta_x + y \beta_y ) }.
\end{align*}

\end{example}

\begin{example}[AFT]\label{ex-3}
The accelerated failure time (AFT) model is given by
\begin{align*}
S_1 ( t \mid x,y) = S_{10} \left( t \exp(x \beta_x + y \beta_y) \right).
\end{align*}

\end{example}

\begin{example}[Piecewise exponential]\label{ex-4}
Consider a piecewise exponential (constant hazard) model on each unit interval with its hazard function given by
\begin{align*}
\lambda_1 ( s \mid x,y,\epsilon ) = \log ( 1 + e^{ - x \beta_x - y \beta_y + f(t) - \epsilon } ), \quad s \in [t-1,t), \; t = 1,2,\ldots,
\end{align*}
where $f(t)$ is a deterministic function of elapsed duration and $\epsilon$ is an individual random effect. For clarity, we denote
\begin{align*}
z_t = x \beta_x + y \beta_y - f(t) + \epsilon, \quad p_t = \frac{1}{1 + e^{z_t}}.
\end{align*}
Then the cumulative hazard is expressed as
\begin{align*}
\Lambda_1 ( t \mid x,y,\epsilon )
&= (t - \lfloor t \rfloor) \lambda_1 ( t \mid x,y,\epsilon ) + \sum_{k = 1}^{\lfloor t \rfloor} \lambda_1 ( k \mid x,y,\epsilon ) \\
&= (t - \lfloor t \rfloor) \log ( 1 + e^{- z_t } ) + \sum_{k = 1}^{\lfloor t \rfloor} \log ( 1 + e^{- z_k } ) \\
&= - (t - \lfloor t \rfloor) \log ( 1 - p_t ) - \sum_{k = 1}^{\lfloor t \rfloor} \log ( 1 - p_k ),
\end{align*}
where $\lfloor x \rfloor$ denotes the largest integer that is less or equal to $x$. Thus, the survival function is
\begin{align*}
S_1 ( t \mid x,y,\epsilon )
&= \exp \left( (t - \lfloor t \rfloor) \log ( 1 - p_t ) + \sum_{k = 1}^{\lfloor t \rfloor} \log ( 1 - p_k ) \right) \nonumber \\
&= ( 1 - p_t )^{t - \lfloor t \rfloor} \prod_{k = 1}^{\lfloor t \rfloor} ( 1 - p_k ).
\end{align*}
Note that $p_t$ is a function of $x \beta_x + y \beta_y$ and so is $S_1 ( t \mid x,y,\epsilon )$.


In fact, the piecewise exponential model corresponds to the discrete duration model in Nicoletti and Rondinelli (2010). To see this, note that for any integer $t$, we have
\begin{align*}
\Pr ( T_1 \geq t \mid T_1 \geq t-1, X = x, Y = y, \mathcal{E} = \epsilon)
&= \frac{\Pr ( T_1 \geq t \mid X = x, Y = y, \mathcal{E} = \epsilon )}{ \Pr ( T_1 \geq t-1 \mid X = x, Y = y, \mathcal{E} = \epsilon ) } \\
&= \frac{ S_1 (t \mid x,y,\epsilon) }{ S_1 (t-1 \mid x,y,\epsilon) } = 1 - p_t.
\end{align*}
Therefore, the discrete hazard in the interval $[t-1,t)$ is
\begin{align*}
&{}\Pr ( T_1 \in [t-1,t) \mid T_1 \geq t-1, X = x, Y = y,\mathcal{E} = \epsilon) \\
&= \Pr ( T_1 < t \mid T_1 \geq t-1, X = x, Y = y,\mathcal{E} = \epsilon) \\
&= 1 - \Pr ( T_1 \geq t \mid T_1 \geq t-1, X = x, Y = y,\mathcal{E} = \epsilon)
= p_t = \frac{1}{1 + e^{z_t}},
\end{align*}
which is the same as Equations~(3)--(4) in Nicoletti and Rondinelli (2010). Note that the random effect $\mathcal{E}$ does not affect the result as long as
\begin{align*}
\frac{\partial}{\partial x} \E_{\mathcal{E}} \left\{ S_1 ( t \mid x,y,\mathcal{E} ) \right\} = \E_{\mathcal{E}} \left\{ \frac{\partial}{\partial x} S_1 ( t \mid x,y,\mathcal{E} ) \right\}.
\end{align*}
\end{example}

It is obvious that the models in Examples~\ref{ex-1}--\ref{ex-4} are included in the single-index model (\ref{eq-si-2}), hence the ratio $\beta_x / \beta_y$ is identifiable and equal to relative covariates effects under the exclusion restriction by Corollary~\ref{cor-si}.

\subsection{Non-single-index model \label{sec:nonSI}}

\begin{example}[Two-Hazards]\label{ex-5} 
We consider a two-hazards additive model in which the marginal hazard for risk 1 is the sum of a time-linear PH term driven by $X$ and a time-constant PH term driven by $Y$:
\begin{align*}
\lambda_1 ( t \mid x,y ) = a_1 e^{x \beta_x} t + b_1 e^{y \beta_y},
\end{align*}
where $a_1,b_1 > 0$ are parameters. The corresponding margin $S_1$ is
\begin{align}\label{non-si}
S_1 ( t \mid x,y ) = \exp \left\{ - \left( \frac{1}{2} a_1 e^{x \beta_x} t^2 + b_1 e^{y \beta_y} t \right) \right\}.
\end{align}
Obviously, the model (\ref{non-si}) does not belong to the single-index model in (\ref{eq-si-2}). By Theorem~\ref{equiv}, we differentiate (\ref{non-si}) with respect to both $x$ and $y$ and take their ratio to obtain
\begin{align*}
\eta_\pi(t,x,y) = \frac{S_{1,x}'(t \mid x,y)}{S_{1,y}'(t \mid x,y)} = \frac{\beta_x}{\beta_y} \frac{a_1  e^{x \beta_x}}{2b_1  e^{y \beta_y}} t \not= \frac{\beta_x}{\beta_y},
\end{align*}
which demonstrates that (\ref{eq-si-1}) does not hold without those conditions in Theorem~\ref{single-index}. In this case, the meaning of the ratio $\beta_x / \beta_y$ is unclear, however, the relative covariate effect $\eta_\pi(t,x,y)$ can always be interpreted through the overall survival $\pi$ based on its own definition in (\ref{proportion}), and through the risk 1 margin $S_1$ based on the relationship in (\ref{margin-proportion-1}) under the exclusion restriction. Since $\eta_\pi (t,x,y)$ depends on $t$ in this model,
\begin{align*}
\eta_\pi (t,x,y) \not= \eta_m (x,y) \not= \frac{m_{1,x}'(x,y)}{m_{1,y}'(x,y)}
\end{align*}
and $\eta_m (x,y)$ can only be interpreted in terms of the overall survival $\pi$ (compare (\ref{mean-proportion})). To see that $\eta_m (x,y)$ is a complicated function of $S_2$ and $C$ in this model, we provide an example. Let
\begin{align*}
S_2 (t) = \exp \left\{ - \left( \frac{1}{2} a_2 t^2 + b_2  t \right) \right\},
\end{align*}
where $a_2,b_2 > 0$ are parameters. Also, let $C$ be the Clayton copula $
C_\theta (s_1,s_2) = (s_1^{-\theta}+s_2^{-\theta}-1)^{-1/\theta}$
with dependence parameter $\theta > 0$. Then, one can evaluate
\begin{align}\label{ex-5-mc}
\eta_m (x,y) = \frac{m_{x}'(x,y)}{m_{y}'(x,y)} = \frac{\displaystyle \int_0^\infty \partial_1 C_\theta (S_1 (t \mid x,y),S_2 (t)) S_{1,x}' (t \mid x,y) dt}{\displaystyle \int_0^\infty \partial_1 C_\theta (S_1 (t \mid x,y),S_2 (t)) S_{1,y}' (t \mid x,y) dt}
\end{align}
numerically and it depends on $a_2$, $b_2$, and $\theta$, i.e., the parameters of $S_2$ and $C_\theta$.

To verify this analytically, we derive $\eta_m (x,y)$ in the limit as $\theta \to 0$, where $C_\theta$ reduces to the independence copula. In this case, we have
\begin{align*}
\pi (t \mid x,y) = \exp \left\{ - \left( A(x) t^2 + B(y)  t \right) \right\},
\end{align*}
where $A(x) = \left( a_1 e^{x \beta_x} + a_2 \right) / 2$ and $B(y) = b_1 e^{y \beta_y} + b_2$. After a lengthy calculation, we get
\begin{align*}
m(x,y) = \int_0^\infty \pi (t \mid x,y) dt = \sqrt{\frac{\pi}{2 A(x)}} \exp \left( \frac{B(y)}{4 A(x)} \right) \mathrm{erfc} \left( \frac{B(y)}{2 \sqrt{A(x)}} \right),
\end{align*}
where
\begin{align*}
\mathrm{erfc} (z) = \frac{2}{\sqrt{\pi}} \int_z^\infty e^{-u^2} du.
\end{align*}
Differentiate $m(x,y)$ with respect to $x$ and $y$ and take their ratio, we have
\begin{align*}
\eta_m (x,y) = \frac{m_{x}'(x,y)}{m_{y}'(x,y)} = \frac{\beta_x}{\beta_y} \frac{a_1 e^{x \beta_x} }{2 b_1 e^{y \beta_y}} \frac{\left(2A(x)+B(y)\right)^2m(x,y)-B(y)}{2A(x) ( 1 - B(y) m(x,y) ) },
\end{align*}
which clearly depends on $a_2$ and $b_2$, i.e., the parameters of $S_2$.

\end{example}


\section{Estimation \label{sec:est}}
The ratio $\eta_\pi (t,x,y)$ is identifiable from $\pi (t \mid x,y)$ and its derivatives, or $\Lambda(t \mid x,y)$ and its derivatives. The ratio $\eta_m (x,y)$ is identifiable from the conditional mean function $m(x,y)$ and its derivatives. The ratio $\eta$ is identifiable from the derivatives of $m(x,y)$ after applying an averaging step over the joint distribution of $(X,Y)$. In this section, we consider appropriate nonparametric estimation approaches for the unknown functionals. Estimation of $\pi (t \mid x,y)$ and its derivatives has already been considered in Hiabu et al.\ (2025). The estimators for $\Lambda(t \mid x,y)$ and $m(x,y)$ and their derivatives are stated below.

The data consist of a random sample $(t_i,x_i,y_i)$ for $i = 1,\ldots,N$. The cumulative overall hazard function is nonparametric and can be recovered from the overall hazard function in the usual way. We therefore first define the nonparametric estimator for $\lambda (t \mid x,y) = f(t \mid x,y)/ \pi(t \mid x,y)$, where $f$ is the density of $T$. A nonparametric estimator is
\begin{equation*}
\hat{\lambda}(t \mid x,y)=\frac{\sum_{i=1}^N\1(t_i=t)K_{h_x}( x-x_i )K_{h_y}( y-y_i )}{\sum_{i=1}^N\1(t_i>t) K_{h_x}( x-x_i )K_{h_y}( y-y_i ) },
\end{equation*}
where $\1 ( \cdot )$ is the indicator function, $K_{h_x} (x) = K(x/h_x)/h_x$ is a Kernel function and $h_x$ is the bandwidth. $K_{h_y} (y)$ is analogously defined. The estimator utilises a product Kernel for the two-dimensional smoothing problem. The estimator for $\Lambda(t \mid x,y)$ is the Nelson-Aalen type estimator
\begin{equation*}
\hat{\Lambda}(t \mid x,y)=\sum_{j|t_j\leq t}\hat{\lambda}(t_j \mid x,y),
\end{equation*}
where the $t_j$s are the distinct realisations of $T$. The estimator for $\Lambda'_x(t \mid x,y)$ is
\begin{multline*}
\hat{\Lambda}'_x(t \mid x,y) \\ = \frac{\left\{\sum_{i=1}^N \1\{t_i> t\} K_{h_x}'( x-x_{i} )K_{h_y}( y-y_{i} )\right\}\left\{\sum_{i=1}^N K_{h_x}( x-x_{i} )K_{h_y}(y-y_{i})\right\}}{\left\{\sum_{i=1}^N K_{h_x}( x-x_{i} )K_{h_y}(y-y_{i}) \right\}^2} \\
- \frac{\left\{\sum_{i=1}^N K_{h_x}'( x-x_{i} )K_{h_y}(y-y_{i}) \right\}\left\{\sum_{i=1}^N \1\{t_i> t\} K_{h_x}( x-x_{i} )K_{h_y}(y-y_{i}) \right\}}{\left\{\sum_{i=1}^N K_{h_x}( x-x_{i} )K_{h_y}(y-y_{i}) \right\}^2} ,
\end{multline*}
where $K_{h}'(x)$ denotes the partial derivative of $K_{h}(x)$ with respect to $x$. $\hat{\Lambda}'_y(t \mid x,y)$ is the estimator for $\Lambda'_y(t \mid x,y)$ and is defined analogously. Once these estimates are available, the ratio $\eta_\pi (t,x,y)$ in (\ref{proportion}) can be estimated by
\begin{align}\label{est-pi}
\hat{\eta}_\pi (t,x,y)= \frac{\hat{\pi}'_x (t \mid x,y) }{\hat{\pi}'_y (t \mid x,y) }
\end{align}
or
\begin{align}\label{est-Lambda}
\hat{\eta}_{\Lambda}(t,x,y)= \frac{\hat{\Lambda}_x' (t \mid x,y)}{\hat{\Lambda}_y' (t \mid x,y)}.
\end{align}
Hiabu et al.\ (2025) state the Kernel estimator for $\hat{\pi}'_x$ and $\hat{\pi}'_y$ and establish weak consistency under appropriate bandwidth and Kernel choice (Hiabu et al., 2025, Corollary 1). The weak consistency of $\hat{\Lambda}_x'$ and $\hat{\Lambda}_y'$ can be established in an analogous way. The estimators $\hat{\eta}_{\pi}$ and $\hat{\eta}_{\Lambda}$ are the averages of the ratio of two weakly consistent nonparametric estimates and therefore inherit this property. These ratios can respond very strongly to finite-sample biases, though. Hiabu et al.\ (2025) study similar ratios of estimators and observe poor finite-sample properties. They introduce trimming rules to improve finite sample behaviour. In our implementation, we also work with trimming rules that restrict the range of $t$ and exclude evaluations of the ratio when the denominator is small.

Subsection \ref{sec:condmean} defines the ratio $\eta_m (x,y)$ on the basis of the derivatives of conditional mean. Recall that the conditional expectation $m(x,y)$ in (\ref{mean}) can be directly estimated by the Nadaraya–Watson estimator which is defined as
\begin{align}\label{nw}
\hat{m} (x,y) = \frac{\sum_{i = 1}^N t_i K_{h_x} (x - x_i)K_{h_y} (y - y_i)}{\sum_{i = 1}^N K_{h_x} (x - x_i)K_{h_y} (y - y_i)}.
\end{align}
The Nadaraya-Watson estimator is known to be consistent under appropriate bandwidth and Kernel function choice. A natural estimator for $m_x' (x,y)$ can be obtained by taking partial derivative of $\hat{m}_x (x,y)$ in (\ref{nw}) with respect to $x$, namely,
\begin{multline}\label{est-dx-mean}
\hat{m}_x' (x,y) =  \frac{ \left\{ \sum_{i = 1}^N t_i K_{h_x}' (x - x_i)K_{h_y} (y - y_i) \right\} \left\{ \sum_{i = 1}^N K_{h_x} (x - x_i)K_{h_y} (y - y_i) \right\} }{ \left\{ \sum_{i = 1}^N K_{h_x} (x - x_i)K_{h_y} (y - y_i) \right\}^2 } \\
- \frac{ \left\{ \sum_{i = 1}^N t_i K_{h_x} (x - x_i)K_{h_y} (y - y_i) \right\} \left\{ \sum_{i = 1}^N K_{h_x}' (x - x_i)K_{h_y} (y - y_i) \right\} }{ \left\{ \sum_{i = 1}^N K_{h_x} (x - x_i)K_{h_y} (y - y_i) \right\}^2 }.
\end{multline}
Similarly, a nonparametric estimator for $m_y' (x,y)$ is
\begin{multline}\label{est-dy-mean}
\hat{m}_y' (x,y) =  \frac{ \left\{ \sum_{i = 1}^N t_i K_{h_x} (x - x_i)K_{h_y}' (y - y_i) \right\} \left\{ \sum_{i = 1}^N K_{h_x} (x - x_i)K_{h_y} (y - y_i) \right\} }{ \left\{ \sum_{i = 1}^N K_{h_x} (x - x_i)K_{h_y} (y - y_i) \right\}^2 } \\
- \frac{ \left\{ \sum_{i = 1}^N t_i K_{h_x} (x - x_i)K_{h_y} (y - y_i) \right\} \left\{ \sum_{i = 1}^N K_{h_x} (x - x_i)K_{h_y}' (y - y_i) \right\} }{ \left\{ \sum_{i = 1}^N K_{h_x} (x - x_i)K_{h_y} (y - y_i) \right\}^2 }.
\end{multline}
By taking the ratio of (\ref{est-dx-mean}) and (\ref{est-dy-mean}), one can estimate $\eta_m(x,y)$ in (\ref{mean-proportion}) by
\begin{align*}
\hat{\eta}_m (x,y) = \frac{\hat{m}_x' (x,y)}{\hat{m}_y' (x,y)} = \frac{S_{t,dx} S(x,y) - S_{t}(x,y) S_{dx}(x,y)}{S_{t,dy}(x,y) S(x,y) - S_{t}(x,y) S_{dy}(x,y)},
\end{align*}
where
\begin{align*}
S_{t,dx}(x,y) &= \sum_{i = 1}^n t_i K_{h_x}' (x - x_i)K_{h_y} (y - y_i), \quad S_{dx}(x,y) = \sum_{i = 1}^n K_{h_x}' (x - x_i)K_{h_y} (y - y_i), \\
S_{t,dy}(x,y) &= \sum_{i = 1}^n t_i K_{h_x} (x - x_i)K_{h_y}' (y - y_i), \quad S_{dy}(x,y) = \sum_{i = 1}^n K_{h_x} (x - x_i)K_{h_y}' (y - y_i), \\
S_t(x,y) &= \sum_{i = 1}^n t_i K_{h_x} (x - x_i)K_{h_y} (y - y_i), \quad S(x,y) = \sum_{i = 1}^n K_{h_x} (x - x_i)K_{h_y} (y - y_i).
\end{align*}
Lastly, the ratio $\eta$ in (\ref{avg-mean-proportion}) can be estimated by using the ratio of the empirical mean of (\ref{est-dx-mean}) and (\ref{est-dy-mean}). The estimator is
\begin{align}\label{new-est2}
\hat{\eta} = \frac{\sum_{i = 1}^n \hat{m}_x' (x_i,y_i)}{\sum_{i = 1}^n \hat{m}_y' (x_i,y_i)}.
\end{align}
When the model is the single-index model in (\ref{eq-si-2}), $\eta_\pi (t,x,y) = \eta_m (x,y) = \eta$ for all $t,x,y$ and it is possible to compare the estimates. $\hat{\eta}$ is expected to be less sensitive to finite-sample biases than the averages of $\hat{\eta}_\pi (t,x,y)$ or $\hat{\eta}_m (x,y)$ over $t,x,y$. This is because $\hat{\eta}$ is constructed as a ratio of averaged estimates in the numerator and denominator, rather than as an average of the ratios.

In summary, $\eta_\pi (t,x,y)$, $\eta_m (x,y)$ and $\eta$ can be estimated without specifying the copula and without specifying the models for $S_1$ and $S_2$. The only other single risk approach that does not require a model for the copula has been developed by Willems et al.\ (2025), but it only partially identifies the model parameters.

Although all estimators for $\eta_\pi (t,x,y)$, $\eta_m (x,y)$ and $\eta$ are consistent under an appropriate bandwidth choice, the introduced approaches likely differ in their finite sample properties, which are studied in the next section with the help of simulations.



\section{Simulations \label{sec:sim}}
This section presents the results of Monte Carlo simulations to assess the finite-sample performance of the suggested approaches. The focus is on the estimation of $\beta_x/\beta_y$. In the first part, the different nonparametric estimators of Section \ref{sec:est} are applied to estimate the relative covariate effect. In the second part, the nonparametric approach is compared to commonly used (semi)parametric single-index models under correct and incorrect specification.

\paragraph{Estimation of $\eta_{\pi}$, $\eta_m$ and $\eta$.}
Several designs for different copulas and $\theta$s are considered. To assess the differences in finite sample behaviour, estimation is done by using the estimators for $\Lambda$, $\pi$ and $m$ and their respective derivatives. The simulation designs are similar to those in Hiabu et al.\ (2025) to allow for comparability. Stata 16MP is used for these simulations.

The margins of both risks are Weibull models with cumulative baseline hazards $\Lambda_{0j}(t)=\lambda_j t^{\eta_j}$ with $(\lambda_1, \eta_1, \lambda_2, \eta_2) = (0.5, 1, 1, 1)$. 
The covariate function for risk 1 is $\exp(x \beta_x + y \beta_y)$ with $(\beta_{x}, \beta_{y})=(1, 1)$. Therefore $S_1(t \mid x,y)=\exp \{ -0.5\exp(x+y) \}$ and $S_2(t)=\exp(-1)$. The copula is either the Clayton copula $C_{\theta}(s_1,s_2)=(s_1^{-\theta}+s_2^{-\theta}-1)^{-1/\theta}$  with $\theta > 0$ or the Gumbel copula $C_{\theta}(s_1,s_2)= \exp [ - \{ ( -\log s_1 )^\theta + ( -\log s_2 )^\theta \}^{ 1/\theta } ]$ with $\theta\in[1,\infty)$. The simulations are done for two different degrees of dependence with Kendall-$\tau\in \{0.1,0.8\}$, where one is rather weak and one is rather strong. The data generating process (DGP) requires the conditional distribution function of the second variable given the first under the specified copula, namely, $F (s_2 \mid s_1) =  \partial_1 C_{\theta}(s_1,s_2)$, as well as its inverse $F^{-1} (v_2 \mid s_1)$. The process of generating the data consists of 4 steps:

\begin{enumerate}
\item Generate two uniform random variables $s_1$ and $v_2$ on $[0,1]$ with $N$ independent random draws.

\item Obtain $N$ realisations of $s_2$ by computing $F^{-1} (v_2 \mid s_1)$:
\begin{itemize}
\item Clayton copula:
\[s_2=\left(1-s_1^{-\theta}+(v_2s_1^{\theta+1})^{-\theta/(\theta+1)}\right)^{-1/\theta}.\]
\item Gumbel copula: $s_2$ is obtained numerically by computing the root of
\begin{align*}
F (s_2 \mid s_1) - v_2 = 0.
\end{align*}
\end{itemize}

\item Randomly generate a sample of size $N$ of independent $x$ and $y$ with marginal distribution $N(0,1)$.  Obtain durations $t_1$ and $t_2$ given $x$ and $y$ by inverting the marginal Weibull survival:
\begin{align*}
t_1 &= \left[-\log(s_1)/\lambda_1^{\eta_1}/\exp(\beta_{x}x+\beta_yy)\right]^{1/\eta_1}, \\
t_2 &= \left[-\log(s_2)/\lambda_2^{\eta_2}\right]^{1/\eta_2}.
\end{align*}

\item Generate observed minimum duration and observed risk by relating $t_1$ and $t_2$.
\end{enumerate}
The DGP results in between $63\%$--$79\%$ of the $t_i$ corresponding to realisations of $T_1$, and the rest to $T_2$, depending on the choice of $C$ and $\tau$. Simulations are conducted for different sample sizes $N=5{,}000$, $N=10{,}000$, and $N=25{,}000$ with $100$ runs. The Epanechnikov Kernel is used for the nonparametric smoothing. The sensitivity of results with respect to the bandwidth choice is checked by setting $h_x = h_y=h$ to either $0.2$ or $0.3$. In this model, $\eta_{\pi}(t,x,y) = \eta_m(x,y) = \eta = \beta_x/\beta_y=1$ for all $t$, $x$, $y$. All presented estimation approaches of Section \ref{sec:est} are therefore applicable. Hiabu et al.\ (2025) provide numerical evidence of the finite-sample performance of a related estimator of a ratio of partial derivatives of $\pi$. Their results suggest that a large dataset of several $10{,}000$ observations is required in an application. This is confirmed by our simulations for the estimation of $\eta_{\pi}$. In our implementation, we evaluate $\hat{\eta}_\pi (t,x,y)$ and $\hat{\eta}_\Lambda (t,x,y)$ in (\ref{est-pi}) and (\ref{est-Lambda}), respectively, on a equidistant grid for $t$ and take the average over these grid points for $x = \bar{x} = \sum_{i = 1}^n x_i / n$ and $y = \bar{y} = \sum_{i = 1}^n y_i / n$, namely,
\begin{align*}
\hat{\eta}_\pi (\bar{x},\bar{y}) &= G^{-1} \sum_{g = 1}^G \hat{\eta}_\pi (t_g,\bar{x},\bar{y}) = G^{-1}\sum_{g=1}^N\frac{\hat{\pi}'_x (t_g \mid \bar{x},\bar{y}) }{\hat{\pi}'_y (t_g \mid \bar{x},\bar{y}) }, \\
\hat{\eta}_\Lambda (\bar{x},\bar{y}) &= G^{-1} \sum_{g = 1}^G \hat{\eta}_\Lambda (t_g,\bar{x},\bar{y}) = G^{-1}\sum_{g=1}^N\frac{\hat{\Lambda}_x' (t_g \mid \bar{x},\bar{y})}{\hat{\Lambda}_y' (t_g \mid \bar{x},\bar{y})},
\end{align*}
where $g=1,\ldots,G$ are the grid points on the support of $t$ with $G=500$. Note that the averaging is not done in the covariates to restrict the local estimation to areas with the greatest marginal densities of $x$ and $y$ and to speed up estimation. As already observed by Hiabu et al.\ (2025), this improves the numerical stability of the results. In a similar fashion, we also evaluate $\hat{\eta}_m (x,y)$ at $x = \bar{x}_n$ and $y = \bar{y}_n$, namely, $\hat{\eta}_m (\bar{x},\bar{y})$, which correspond to the center of mass of the covariates.

The simulation results for $\hat{\eta}$ are given in Table~\ref{T:sim_res5}. The mean of the realisations of $\hat{\eta}$ is very close to the true value for all sample sizes. The percentiles of the distribution of estimates get closer to the true values when $N$ increases. This shows desirable properties. At the same time, the intervals get rather wide for $N=5,000$ which suggests that this approach requires at least several thousand observations. Taking the average before division instead of averaging the ratios has therefore considerably improved the finite-sample performance. There is no need to choose values of $t$, $x$ and $y$ for the estimation because the estimator averages over all observations.

\begin{table}[!htbp]
\centering
\scriptsize
\caption{Simulation results for $\hat{\eta}$ using $\sum_i \hat{m}_x' (x_i,y_i)/\sum_i\hat{m}_y' (x_i,y_i)$ under the Gumbel or Clayton copula with $\tau = 0.1$ or $\tau = 0.8$. True $\eta=1$.}
\begin{tabular}{rccccc}
\hline\hline
& \multicolumn{2}{c}{Gumbel} & & \multicolumn{2}{c}{Clayton}\\
& $h=0.2$ & $h=0.3$ & & $h=0.2$ & $h=0.3$\\
\hline \\
& \multicolumn{5}{c}{$\tau=0.1$}\\
\multicolumn{3}{l}{$N=5,000$}\\[0.3cm]
mean $\hat{\eta}$ &  $0.9867$ & $0.9918$  &  & $0.9886$  &  $0.9860$   \\
5th, 95th percentile of $\hat{\eta}$ & $[0.7263,1.2303]$ & [0.7824,1.2382]&  & [0.7062,1.2499] & $[0.7567,1.2159]$  \\[0.3cm]
\multicolumn{3}{l}{$N=10,000$}\\[0.3cm]
mean $\hat{\eta}$ &  $1.0015$ & $0.9984$  &  & $0.9995$  &  $1.0056$   \\
5th, 95th percentile of $\hat{\eta}$ & $[0.8342,1.2217]$ & [0.8483,1.2278]&  & [0.8050,1.2290] & $[0.8344,1.2056]$   \\[0.3cm]
\multicolumn{3}{l}{$N=25,000$}\\[0.3cm]
mean $\hat{\eta}$ &  $1.0097$ & $1.0110$  &  & $0.9806$  &  $0.9823$   \\
5th, 95th percentile of $\hat{\eta}$ & $[0.8972,1.1267]$ & [0.9058,1.1359]&  & [0.8761,1.1122] & $[0.8780,1.1127]$ \\[0.5cm]
& \multicolumn{5}{c}{$\tau=0.8$}\\
\multicolumn{3}{l}{$N=5,000$}\\[0.3cm]
mean $\hat{\eta}$ &  $0.9496$ & $0.9676$  &  & $1.0136$  &  $1.0095$   \\
5th, 95th percentile of $\hat{\eta}$ & $[0.5432,1.5528]$ & [0.6242,1.5534]&  & [0.5575,1.5842] & $[0.5322,1.5892]$  \\[0.3cm]
\multicolumn{3}{l}{$N=10,000$}\\[0.3cm]
mean $\hat{\eta}$ &  $1.0063$ & $1.0044$  &  & $0.9990$  &  $1.0103$   \\
5th, 95th percentile of $\hat{\eta}$ & $[0.7037,1.3765]$ & [0.7074,1.4080]&  & [0.6684,1.4482] & $[0.7028,1.4590]$   \\[0.3cm]
\multicolumn{3}{l}{$N=25,000$}\\[0.3cm]
mean $\hat{\eta}$ &  $1.0267$ & $1.0332$  &  & $0.9813$  &  $0.9821$   \\
5th, 95th percentile of $\hat{\eta}$ & $[0.7714,1.2865]$ & [0.7857,1.2700]&  & [0.7691,1.2094] & $[0.7719,1.2314]$ \\
\hline
\hline
\end{tabular}
\label{T:sim_res5}
\end{table}


The results for the estimators $\hat{\eta}_{\pi} (\bar{x},\bar{y})$, $\hat{\eta}_{\Lambda} (\bar{x},\bar{y})$ and $\hat{\eta}_m (\bar{x},\bar{y})$ are given in Supplementary Material. When the estimation uses $\pi$ or $\Lambda$ and their derivatives, estimation is done on an equidistant grid for $t$ with 500 points on $[0.04,3.55]$. In these cases, appropriate trimming rules are applied to avoid instability of results when there are low observation numbers in some areas of the support or the denominators are very small. No trimming in the $t$ dimension is required when estimation is based on the conditional mean function. Tables \ref{T:sim_res1}--\ref{T:sim_res4} present the results for different bandwidths $h$, $\tau$ and copulas. It is apparent that these approaches require several 10,000 observations to get some stable results and that $\hat{\eta}$ has the best performance.

A selection of estimated nonparametric $\hat{\pi}$ and $\hat{\Lambda}$ and their derivatives are shown in Figure~\ref{F:sim1} to illustrate the role of the bandwidth choice. For $h=0.2$, the bias is smaller but the percentiles of the distribution are wider spread, showing the greater variance of the estimates. Figure~\ref{F:sim2} shows results for the distribution of the estimated $\hat{\eta}_{\pi}(t,\bar{x},\bar{y})$ as a function of $t$. It is apparent that in some areas of the support of $t$, the estimates are very noisy. A similar observation has already been made in Hiabu et al. (2025). To address this and improve the quality of the estimates, we trim the support of $t$ to $[0.05,2.5]$ in the case of the Gumbel copula and $[0.05,0.5]$ in the case of Clayton. By doing so, we trim the right end of the $t$ where estimates are found to be more noisy mainly due to less observations. This is denoted as trimming I: boundary trimming in the tables and is made visible by the dotted vertical lines in Figure \ref{F:sim2} for the Gumbel copula. An additional trimming rule ignores grid points for which denominator in $\hat{\eta}_{\pi}(t,\bar{x},\bar{y})$ is too close to zero. This is when $\hat{\Lambda}_y$ or $\hat{\pi}_y$ are less than $0.1$ as a small estimation bias in this estimates leads to a sizable bias of $\hat{\eta}_{\pi}(\bar{x},\bar{y})$ and $\hat{\eta}_{\Lambda}(\bar{x},\bar{y})$. This is denoted as trimming rule II: boundary and denominator trimming in the tables. Estimation on the grounds of the conditional mean functions does not require this trimming.

To sum up, the simulation results confirm that the suggested estimators converge to their true values and their variances decrease as the sample size $N$ increases. They also show that the purely nonparametric approach require samples of ideally at least several thousand observations. The preferable approach is the estimator $\hat{\eta}$ in (\ref{new-est2}), which divides the sample averages of nonparametric mean regression estimates. The approaches that average the ratio of two nonparametric estimates have considerably worse properties and require 10,000s of observations.

\paragraph{Comparison with (semi)parametric single-index models.}
The previous results suggested that the nonparametric approaches lack the precision of (semi)parametric approaches under correct specification and require larger data sets. Their advantage is that they are consistent for $\eta = \E[m_x'(X,Y)] / \E[m_y'(X,Y)]$ in (\ref{avg-mean-proportion}), but not for $\beta_x / \beta_y$ under weaker restrictions, in particular when the model is not a single-index model or when the single-index component is misspecified. In the following, we compare the performance of the estimated $\eta$ on the grounds of the nonparametric estimator in (\ref{new-est2}) and the ratio of estimated coefficients of the single-index models: the semiparametric Cox PH model, the AFT model with Weibull distribution and the PO model (compared Section \ref{sec:examples}). When the single-index model is correctly specified, it is more efficient and therefore preferable over the nonparametric approach. For non-single-index models, the (semi)parametric and the nonparametric cannot be used to consistently estimate $\beta_x / \beta_y$, although the nonparametric approach still consistently estimates $\eta$. When the single-index component is misspecified, the nonparametric approach is consistent for $m_x'(x,y) / m_y'(x,y)=\eta_{\pi}(x,y)$ and $\eta$, but not for $\beta_x / \beta_y$ as $\eta_{\pi} (x,y)$ depends on $x$ and $y$. We apply the DGP of the Weibull model of the previous simulations as the correctly specified model. The DGPs that violates the single-index property is the Two-hazards model of Example \ref{ex-5} in Subsection \ref{sec:nonSI} with $(a_1,b_1,a_2,b_2) = (1,1,0.5,1)$ and $(\beta_{x},\beta_{y})=(1,1)$. The model with misspecified single-index is the Weibull model with covariate function $\exp(\beta_x+\beta_y+\beta_{x^2}x^2)$ for risk 1 with $(\beta_{x},\beta_{y},\beta_{x^2})=(1,1,2)$, but the covariate function $\exp(\beta_x+\beta_y)$ is used for the estimation of the (semi)parametric models. The simulations are run for $N=5{,}000$ and $25{,}000$ and use the Clayton copula with Kendall-$\tau = 0.8$. The bandwidth is selected by 10-fold Cross-Validation. The statistical software R V.4.5.2 is used for these simulations.

\begin{table}[!htbp]
\centering
\scriptsize
\caption{Simulation results for $\hat{\eta}$ under the Clayton copula with $\tau=0.8$. Weibull AFT is a single-index model, while Two-Hazards is a non-single-index model and Weibull AFT $x^2$ is a single-index model with misspecified covariate function.}
\begin{tabular}{rccc}
\hline \hline
Estimator\hspace{0.3cm}\textbackslash \hspace{0.3cm}DGP & (a) & (b) & (c) \\
                                                        & Weibull AFT &  Two-Hazards & Weibull AFT $x^2$ \\
                                                        & $\eta=\beta_x/\beta_y=1$ & $\eta \neq \beta_x/\beta_y = 1$ & $\eta \neq \beta_x/\beta_y=1$ \\
\hline
\multicolumn{4}{l}{$N=5,000$}\\[0.5cm]
\multicolumn{1}{l}{Nonparametric}\\
mean $\hat{\eta}$ &  $0.9971$ &    $0.6487$ &   $0.6104$   \\
5th, 95th percentile of $\hat{\eta}$ & $[0.7393,1.3452]$ &  $[0.4507,0.8619]$ &   $[0.4391,0.7776]$  \\[0.3cm]
\multicolumn{1}{l}{Semiparametric Cox PH}\\
mean $\hat{\eta}$ &  $0.9952$ &    $0.6254$ &   $1.0442$   \\
5th, 95th percentile of $\hat{\eta}$ & $[0.9548,1.0397]$ &  $[0.5587,0.6872]$ &   $[0.8968,1.2069]$  \\[0.3cm]
\multicolumn{1}{l}{Parametric Weibull AFT}\\
mean $\hat{\eta}$ &  $0.9952$ &    $0.6137$ &   $1.0489$   \\
5th, 95th percentile of $\hat{\eta}$ & $[0.9539,1.0429]$ &  $[0.5494,0.6727]$ &   $[0.9003,1.2128]$  \\[0.3cm]
\multicolumn{1}{l}{Semiparametric PO}\\
mean $\hat{\eta}$ &  $0.9935$ &    $0.5693$ &   $0.8861$   \\
5th, 95th percentile of $\hat{\eta}$ & $[0.9454,1.0429]$ &  $[0.5123,0.6278]$ &   $[0.7566,1.0346]$  \\[0.5cm]
\multicolumn{4}{l}{$N=25,000$}\\[0.5cm]
\multicolumn{1}{l}{Nonparametric}\\
mean $\hat{\eta}$ &  $1.0079$ &    $0.6713$ &  $0.5879$   \\
5th, 95th percentile of $\hat{\eta}$ & $[0.8982,1.1312]$ &  $[0.5925,0.7781]$ &    $[0.5124,0.6534]$  \\[0.3cm]
\multicolumn{1}{l}{Semiparametric Cox PH}\\
mean $\hat{\eta}$ &  $1.0006$ &    $0.6330$ &   $1.0559$   \\
5th, 95th percentile of $\hat{\eta}$ & $[0.9803,1.0181]$ &  $[0.6044,0.6584]$ &   $[0.9909,1.1310]$  \\[0.3cm]
\multicolumn{1}{l}{Parametric Weibull AFT}\\
mean $\hat{\eta}$ &  $1.0008$ &    $0.6204$ &   $1.0611$   \\
5th, 95th percentile of $\hat{\eta}$ & $[0.9797,1.0183]$ &  $[0.5927,0.6459]$ &   $[0.9946,1.1356]$  \\[0.3cm]
\multicolumn{1}{l}{Semiparametric PO}\\
mean $\hat{\eta}$ &  $0.9997$ &    $0.5762$ &   $0.8972$   \\
5th, 95th percentile of $\hat{\eta}$ & $[0.9785,1.0185]$ &  $[0.5513,0.5996]$ &   $[0.8363,0.9629]$  \\[0.3cm]
\hline \hline
\end{tabular}
\label{T:sim_res6}
\end{table}

The results are shown in Table \ref{T:sim_res6}. Column (a) is the correctly specified single-index model. Column (b) is the non-single-index model and column (c) is the model with the misspecified single-index. It is apparent from (a) that all methods are consistent for $\eta$ and the conventional (semi)parametric models are more efficient when the underlying model is a correctly specified single-index model. It illustrates that the robustness property of the single-index models (Li and Duan, 1989) may carry over to settings with (dependent) censoring or competing risks. The results in (b) show that all models do not consistently estimate the ratio of parameters and they converge to different points except the Cox PH and the Weibull AFT models which give very similar results. The nonparametric approach converges in probability to $\eta = \E[m_x'(X,Y)] / \E[m_y'(X,Y)] = 0.6759$ in this setup, where the true value is approximated numerically using the Monte Carlo method to evaluate the expected values of the numerator and denominator of (\ref{ex-5-mc}).  The previous simulations have already confirmed that the nonparametric estimator for $\eta_\pi(x,y)$ is consistent and we omit the reporting. The results in (c) for the misspecified covariate function in a single-index model show a similar pattern. The true value $\eta = 0.5617$ in this setup is approximated numerically as well. The fact that the different estimators converge to different points when the model is not a correctly specified single-index model, provides evidence against this assumption. A bootstrap-based test can be utilised to detect systematic differences in the $\hat{\eta}$s of the different approaches.

\section{Empirical applications \label{sec:app}}
We present two applications to put our framework into practice. One data set is from economics to study unemployment duration, the other data set is medical data. It contains patient survival times in hospitals. For both applications, we first obtain the nonparametric estimator $\hat{\eta}$ in (\ref{new-est2}). Then, we also compute the ratios of regression parameters
\begin{align*}
\hat{\eta}_{\text{Cox}} = \frac{\hat{\beta}_x^\text{Cox}}{\hat{\beta}_y^\mathrm{Cox}}, \quad \hat{\eta}_{\text{PO}} = \frac{\hat{\beta}_x^\text{PO}}{\hat{\beta}_y^\text{PO}},
\end{align*}
where $\hat{\beta}_x^\text{Cox}$ and $\hat{\beta}_y^\text{Cox}$ ($\hat{\beta}_x^\text{PO}$ and $\hat{\beta}_y^\text{PO}$) are estimators for the regression parameters $\beta_x^\text{Cox}$ and $\beta_y^\text{Cox}$ ($\beta_x^\text{PO}$ and $\beta_y^\text{PO}$) under the Cox (PO) model, respectively. Under the exclusion restriction in (\ref{copula}), we consider the hypotheses
\begin{align*}
H_0: \eta = \frac{\beta_x^\mathrm{Cox}}{\beta_y^\mathrm{Cox}}~(\text{i.e., the Cox PH model holds})\quad &\text{versus} \quad H_1: \eta \not= \frac{\beta_x^\mathrm{Cox}}{\beta_y^\mathrm{Cox}},\\
H_0: \eta = \frac{\beta_x^\mathrm{PO}}{\beta_y^\mathrm{PO}}~(\text{i.e., the PO model holds})\quad &\text{versus} \quad H_1: \eta \not= \frac{\beta_x^\mathrm{PO}}{\beta_y^\mathrm{PO}}.
\end{align*}
These are tested with a bootstrap-based test with 400 repetitions, where the p-value is obtained from the nonparametric bootstrap distribution of the statistic. The bandwidth for the nonparametric estimation is chosen by 5-fold cross-validation.

\paragraph{Analysis of unemployment} We use the data set on unemployment duration. Risk 1 is the time (in days) until an unemployed person starts a job. End of data censoring and all other exit states such as out of the labour force (inactivity for various reasons) and self-employment are pooled into one censoring variable. It would be unrealistic to assume that time to start a job and the censoring time are independent. The sample is an extract of the IAB-Employment Sample 1975-2001 (IABS-R01) and is described in more detail in Wichert and Wilke (2008). It contains 21,685 observations, of which 3,017 (14.16\%) are censored. The covariate $X$ is age (in years). The covariate $Y$ is daily pre-unemployment wage (in EUR).

We obtain the following estimates: Nonparametric $\hat{\eta} = -3.8994$, Semiparametric Cox PH $\hat{\eta}_{\text{Cox}} = -4.0310$ and PO $\hat{\eta}_{\text{PO}} = -2.9920$. While the first two are similar, the PO estimate strongly differs from the first two. The p-value for testing $H_0: \eta = \beta_x^\mathrm{Cox} / \beta_y^\mathrm{Cox}$ is $0.725$, whereas the p-value for testing $H_0: \eta = \beta_x^\mathrm{PO}/\beta_y^\mathrm{PO}$ is $0.000$. The PO model is therefore clearly rejected, while there is no evidence against the Cox model. Note that $\hat{\eta}$ is a consistent estimate of $\E[m_x'(X,Y)] / \E[m_y'(X,Y)]$ and when the Cox model is not rejected, it becomes a consistent estimate of $\beta_x^\mathrm{Cox}/\beta_y^\mathrm{Cox}$. It means that one partial covariate effect on $S_1$ is positive, while the other is negative. The estimated covariate effect of $age$ is around 4-times larger in size than that of $wage$.

\paragraph{Analysis of time to death}
We use an extract of the Study to Understand Prognoses Preferences Outcomes and Risks of Treatment (SUPPORT) of Vanderbilt University which is made available by Bhatnagar et al.~(2022). The data contains time to death for 9,104 hospital patients, of which 2,904 (31.9\%) are censored. The covariate $X$ is age (in years). The covariate $Y$ is SPS (SUPPORT physiology score). We obtain the following estimates: Nonparametric $\hat{\eta} = 0.3660$, Semiparametric Cox PH $\hat{\eta}_{\text{Cox}} = 0.3243$ and PO $\hat{\eta}_{\text{PO}} = 0.2598$. The p-value for testing $H_0: \eta = \beta_x^\mathrm{Cox}/\beta_y^\mathrm{Cox}$ is $0.100$, whereas the p-value for testing $H_0: \eta = \beta_x^\mathrm{PO}/\beta_y^\mathrm{PO}$ is $0.000$. Here once again, there is strong evidence against the PO model. While the Cox model cannot be rejected at high levels, there is some weak evidence against it as well. Because $\hat{\eta}$ is positive, the sign of the partial covariate effects on $S_1$ is the same in this application.

The two applications have demonstrated that with samples of 10K or 20K observations, it is possible to reject the specification of popular (semi)parametric single-index models in situations with unknown dependent censoring or competing risks without that the model is identified. Because the Cox PH model is not rejected in both cases, our approach gives consistent estimates of the ratio of partial covariate effects on $S_1$, despite that $S_1$ is not identifiable.

\section*{Funding statement}
The authors declare that no external funding has been received to conduct the study.

\section*{Data availability statement}
The unemployment data are an extract of the IAB-Employment Sample 1975-2001 (IABS-R01) which is managed by the Institute for Employment Research (\url{https://fdz.iab.de}). The same sample has been used by Wichert and Wilke (2008) and has been made available as a public-use file by the data provider. The SUPPORT data is available through the R-package \texttt{casebase} by Bhatnagar et al. (2022). It is a subsample of the Study to Understand Prognoses Preferences Outcomes and Risks of Treatment (SUPPORT) of the Department of Biostatistics at Vanderbilt University. Upon acceptance of the paper, the estimation samples and R-code to replicate the results of the application will be made available for downloaded from: \url{https://github.com/ralfawilke/nonparreg}.

\clearpage
	
\setcounter{page}{1}

\renewcommand{\appendixpagename}{\centering Nonparametric regression with dependent censoring or competing risks \\ \centering SUPPLEMENTARY MATERIAL}
\begin{appendices}

\appendixpage
\centering
\author{
Jia-Han Shih\footnote{National Sun Yat-sen University, Department of Applied Mathematics, 80424, Kaohsiung, Taiwan, E--mail: jhshih@math.nsysu.edu.tw},
Simon M.S. Lo\footnote{United Arab Emirates University, Department of Economics and Finance, E--mail: losimonms@yahoo.com.hk},
Ralf A. Wilke\footnote{Copenhagen Business School, Department of Economics, Porcel{\ae}nshaven 16A, 2000 Frederiksberg, DK, E--mail: rw.eco@cbs.dk}
}

\vspace{4cm}

\thispagestyle{empty}
\pagebreak
			
\renewcommand{\thesection}{S\arabic{section}}
\renewcommand{\thetable}{S\arabic{table}}
\renewcommand{\thefigure}{S\arabic{figure}}

\pagebreak
\linespread{1.3}


\begin{landscape}
\begin{table}[!htbp]
\centering
\scriptsize
\caption{Simulation results for $\hat{\eta}_{\pi}$, $\hat{\eta}_{\Lambda}$ and $\hat{\eta}_m$ under the Gumbel copula with $\tau = 0.1$. True $\eta = 1$.}
	\begin{tabular}{rccccccc}
		\hline\hline
& \multicolumn{3}{c}{$h=0.2$} & & \multicolumn{3}{c}{$h=0.3$}\\
& no trimming & trimming I & trimming II & & no trimming & trimming I & trimming II  \\
\hline
\multicolumn{3}{l}{$N=10,000$}\\[0.3cm]
mean $\hat{\eta}_{\pi}(\bar{x},\bar{y})$ &  $1.1612$ & $1.0434$   & $1.0848$ &  &  $0.3028$   &  $-0.0774$     &  $1.2360$  \\
5th, 95th percentile of $\hat{\eta}_{\pi}(\bar{x},\bar{y})$ & $[-3.6390,5.2424]$ & $[-2.4787,4.8741]$ & $[-0.2044,2.5523]$ & & $[0.4177,2.5644]$ & $[0.5177,2.7999]$ & $[0.5781,2.5989]$\\
mean $\hat{\eta}_{\Lambda}(\bar{x},\bar{y})$ &  $-0.0939$ & $1.6523$   & $1.1480$ &  &  $1.1456$   &  $1.4206$     &  $1.2623$  \\
5th, 95th percentile of $\hat{\eta}_{\Lambda}(\bar{x},\bar{y})$ & $[-4.2138,9.9594]$ & $[-0.9041,12.5587]$ & $[-0.2532,3.2753]$ & & $[0.4169,2.5929]$ & $[0.5177,2.7342]$ & $[0.5796,2.7464]$\\
mean $\hat{\eta}_{m}(\bar{x},\bar{y})$ &  $0.1823$ &   &  &  &  $1.1060$   &       &    \\
5th, 95th percentile of $\hat{\eta}_{m}(\bar{x},\bar{y})$ & $[-0.1735,4.0250]$ &&  & & $[0.5636,1.9778]$ &  & \\[0.3cm]
\multicolumn{3}{l}{$N=25,000$}\\[0.3cm]
mean $\hat{\eta}_{\pi}(\bar{x},\bar{y})$ &  $1.3395$ & $0.8202$   & $1.2633$ &  &  $1.1081$   &  $1.0752$     &  $1.0767$  \\
5th, 95th percentile of $\hat{\eta}_{\pi}(\bar{x},\bar{y})$ & $[-0.9655,4.8839]$ & $[-0.4928,2.8807]$ & $[0.4517,2.8822]$ & & $[0.7093,1.5450]$ & $[0.7471,1.4623]$ & $[0.7471,1.4623]$\\
mean $\hat{\eta}_{\Lambda}(\bar{x},\bar{y})$ &  $-0.6871$ & $1.2128$   & $1.3055$ &  &  $1.1077$   &  $1.0743$     &  $1.0768$  \\
5th, 95th percentile of $\hat{\eta}_{\Lambda}(\bar{x},\bar{y})$ & $[-3.0534,2.8380]$ & $[-0.8468,3.4866]$ & $[0.4484,2.7232]$ & & $[0.7100,1.5491]$ & $[0.7466,1.4644]$ & $[0.7466,1.4644]$\\
mean $\hat{\eta}_{m}(\bar{x},\bar{y})$ &  $1.1824$ &   &  &  &  $1.0499$   &       &    \\
5th, 95th percentile of $\hat{\eta}_{m}(\bar{x},\bar{y})$ & $[0.4294,2.3902]$ &&  & & $[0.7311,1.4227]$ &  & \\[0.3cm]
\multicolumn{3}{l}{$N=50,000$}\\[0.3cm]
mean $\hat{\eta}_{\pi}(\bar{x},\bar{y})$ &  $1.094$ & $0.9392$   & $1.0553$ &  &  $1.0381$   &  $1.0269$     &  $1.0266$  \\
5th, 95th percentile of $\hat{\eta}_{\pi}(\bar{x},\bar{y})$ & $[0.5839,2.2877]$ & $[0.5560,1.5551]$ & $[0.5529,1.5346]$ & & $[0.8143,1.3474]$ & $[0.7842,1.2725]$ & $[0.7842,1.2725]$\\
mean $\hat{\eta}_{\Lambda}(\bar{x},\bar{y})$ &  $1.2737$ & $1.1303$   & $1.0617$ &  &  $1.0390$   &  $1.0269$     &  $1.0266$  \\
5th, 95th percentile of $\hat{\eta}_{\Lambda}(\bar{x},\bar{y})$ & $[0.5867,2.6366]$ & $[0.5032,1.6132]$ & $[0.5889,1.5563]$ & & $[0.8148,1.3462]$ & $[0.7830,1.2725]$ & $[0.7830,1.2725]$\\
mean $\hat{\eta}_{m}(\bar{x},\bar{y})$ &  $0.9865$ &   &  &  &  $1.0141$   &       &    \\
5th, 95th percentile of $\hat{\eta}_{m}(\bar{x},\bar{y})$ & $[0.5408,1.4277]$ &&  & & $[0.7782,1.2640]$ &  & \\[0.3cm]
\hline
\hline
\multicolumn{7}{l}{\footnotesize Note: trimming I: boundary trimming, trimming II: boundary and denominator trimming}
\end{tabular}
\label{T:sim_res1}
\end{table}
\end{landscape}

\begin{landscape}
\begin{table}[!htbp]
\centering
\scriptsize
\caption{Simulation results for $\hat{\eta}_{\pi}$, $\hat{\eta}_{\Lambda}$ and $\hat{\eta}_m$ under the Gumbel copula with $\tau = 0.8$. True $\eta = 1$.}
	\begin{tabular}{rccccccc}
		\hline\hline
& \multicolumn{3}{c}{$h=0.2$} & & \multicolumn{3}{c}{$h=0.3$}\\
& no trimming & trimming I & trimming II & & no trimming & trimming I & trimming II  \\
\hline
\multicolumn{3}{l}{$N=10,000$}\\[0.3cm]
mean $\hat{\eta}_{\pi}(\bar{x},\bar{y})$ &  $1.8463$ & $0.8165$   & $0.2957$ &  &  $-1.7886$   &  $0.2533$     &  $0.7338$  \\
5th, 95th percentile of $\hat{\eta}_{\pi}(\bar{x},\bar{y})$ & $[-6.2924,8.0688]$ & $[-4.1398,6.7546]$ & $[-1.6802,1.7200]$ & & $[-3.4945,3.5857]$ & $[-3.8297,4.7962]$ & $[-0.9383,2.6051]$\\
mean $\hat{\eta}_{\Lambda}(\bar{x},\bar{y})$ &  $-0.0690$ & $0.1244$   & $0.2620$ &  &  $0.8915$   &  $1.0380$     &  $0.7002$  \\
5th, 95th percentile of $\hat{\eta}_{\Lambda}(\bar{x},\bar{y})$ & $[-5.2989,5.8529]$ & $[0.1245,7.0039]$ & $[-1.9730,2.1788]$ & & $[-2.5079,5.4999]$ & $[-2.7333,6.5770]$ & $[-1.6476,3.1162]$\\
mean $\hat{\eta}_{m}(\bar{x},\bar{y})$  &  $1.3417$ &   &  &  &  $1.2613$   &       &    \\
5th, 95th percentile of $\hat{\eta}_{m}(\bar{x},\bar{y})$ & $[-4.6270,5.4684]$ &&  & & $[-4.4760,3.1365]$ &  & \\[0.3cm]
\multicolumn{3}{l}{$N=25,000$}\\[0.3cm]
mean $\hat{\eta}_{\pi}(\bar{x},\bar{y})$  &  $5.4444$ & $6.5357$   & $0.1612$ &  &  $0.0897$   &  $-0.5369$     &  $1.0007$  \\
5th, 95th percentile of $\hat{\eta}_{\pi}(\bar{x},\bar{y})$ & $[-3.5916,10.7829]$ & $[-4.7560,8.2996]$ & $[-1.6449,2.0648]$ & & $[-12.1987,8.8592]$ & $[-17.3816,6.6042]$ & $[-0.4238,2.4764]$\\
mean $\hat{\eta}_{\Lambda}(\bar{x},\bar{y})$ &  $-2.2828$ & $-3.2733$   & $0.1717$ &  &  $0.1441$   &  $1.4380$     &  $1.0414$  \\
5th, 95th percentile of $\hat{\eta}_{\Lambda}(\bar{x},\bar{y})$  & $[-10.2732,6.6499]$ & $[-7.7828,6.7319]$ & $[-1.8503,2.6629]$ & & $[-4.2799,6.8051]$ & $[-4.006,7.0498]$ & $[-0.3106,2.4160]$\\
mean $\hat{\eta}_{m}(\bar{x},\bar{y})$  &  $0.3970$ &   &  &  &  $1.1902$   &       &    \\
5th, 95th percentile of $\hat{\eta}_{m}(\bar{x},\bar{y})$ & $[-9.6467,9.0209]$ &&  & & $[-0.2960,4.5731]$ &  & \\[0.3cm]
\multicolumn{3}{l}{$N=50,000$}\\[0.3cm]
mean $\hat{\eta}_{\pi}(\bar{x},\bar{y})$ &  $0.9010$ & $1.1743$   & $0.4268$ &  &  $2.9524$   &  $3.8444$     &  $1.2234$  \\
5th, 95th percentile of $\hat{\eta}_{\pi}(\bar{x},\bar{y})$ & $[-3.6684,9.4744]$ & $[-3.2962,10.8197]$ & $[-1.2198,2.0436]$ & & $[0.04832,5.8561]$ & $[-0.0564,8.0571]$ & $[0.3411,3.5164]$\\
mean $\hat{\eta}_{\Lambda}(\bar{x},\bar{y})$ &  $-4.2229$ & $-6.9371$   & $0.4447$ &  &  $1.1079$   &  $1.1601$     &  $1.2999$  \\
5th, 95th percentile of $\hat{\eta}_{\Lambda}(\bar{x},\bar{y})$ & $[-14.0540,10.9535]$ & $[-12.4526,8.3478]$ & $[-1.5348,2.2304]$ & & $[-0.3769,4.8765]$ & $[-1.2778,5.9224]$ & $[0.3451,3.6447]$\\
mean $\hat{\eta}_{m}(\bar{x},\bar{y})$ &  $4.0054$ &   &  &  &  $1.2852$   &       &    \\
5th, 95th percentile of $\hat{\eta}_{m}(\bar{x},\bar{y})$ & $[-3.2487,7.1045]$ &&  & & $[0.3108,3.6897]$ &  & \\[0.3cm]
\hline
\hline
\multicolumn{7}{l}{\footnotesize Note: trimming I: boundary trimming, trimming II: boundary and denominator trimming}
\end{tabular}
\label{T:sim_res2}
\end{table}
\end{landscape}

\begin{landscape}
\begin{table}[!htbp]
\centering
\scriptsize
\caption{Simulation results for $\hat{\eta}_{\pi}$, $\hat{\eta}_{\Lambda}$ and $\hat{\eta}_m$ under the Clayton copula with $\tau = 0.1$. True $\eta = 1$.}
	\begin{tabular}{rccccccc}
		\hline\hline
& \multicolumn{3}{c}{$h=0.2$} & & \multicolumn{3}{c}{$h=0.3$}\\
& no trimming & trimming I & trimming II & & no trimming & trimming I & trimming II  \\
\hline
\multicolumn{3}{l}{$N=10,000$}\\[0.3cm]
mean $\hat{\eta}_{\pi}(\bar{x},\bar{y})$ &  $-4.35780$ & $0.0746$   & $0.0746$ &  &  $1.8866$   &  $0.8291$     &  $0.8291$  \\
5th, 95th percentile of $\hat{\eta}_{\pi}(\bar{x},\bar{y})$ & $[-6.5505,4.6196]$ & $[-2.1456,2.0357]$ & $[-2.1456,2.0357]$ & & $[-5.5474,4.2420]$ & $[-0.8960,3.0681]$ & $[-.8961,3.0681]$\\
mean $\hat{\eta}_{\Lambda}(\bar{x},\bar{y})$ &  $0.5755$ & $-0.0319$   & $0.0502$ &  &  $5.3327$   &  $1.2241$     &  $0.8668$  \\
5th, 95th percentile of $\hat{\eta}_{\Lambda}(\bar{x},\bar{y})$ & $[-3.8227,6.6187]$ & $[-7.9394,7.5409]$ & $[-2.1155,2.1458]$ & & $[-2.7382,5.0315]$ & $[-4.2073,4.6337]$ & $[-1.1874,3.4027]$\\
mean $\hat{\eta}_{m}(\bar{x},\bar{y})$ &  $-1.5885$ &   &  &  &  $1.0709$   &       &    \\
5th, 95th percentile of $\hat{\eta}_{m}(\bar{x},\bar{y})$ & $[-8.6183,6.1173]$ &&  & & $[-2.8726,5.6834]$ &  & \\[0.3cm]
\multicolumn{3}{l}{$N=25,000$}\\[0.3cm]
mean $\hat{\eta}_{\pi}(\bar{x},\bar{y})$ &  $0.9709$ & $0.5195$   & $0.5195$ &  &  $1.2518$   &  $1.1416$     &  $1.1416$  \\
5th, 95th percentile of $\hat{\eta}_{\pi}(\bar{x},\bar{y})$ & $[-4.2450,5.2034]$ & $[-1.2936,2.2908]$ & $[-1.2936,2.2908]$ & & $[-0.8173,4.0099]$ & $[0.1866,2.5247]$ & $[0.1866,2.5247]$\\
mean $\hat{\eta}_{\Lambda}(\bar{x},\bar{y})$ &  $-1.3035$ & $-0.2664$   & $0.4929$ &  &  $1.3561$   &  $1.2732$     &  $1.1382$  \\
5th, 95th percentile of $\hat{\eta}_{\Lambda}(\bar{x},\bar{y})$ & $[-7.2909,3.2007]$ & $[-11.2414,9.5439]$ & $[-1.6001,2.1170]$ & & $[-0.3091,3.8454]$ & $[-0.1232,3.3639]$ & $[0.1879,2.5540]$\\
mean $\hat{\eta}_{m}(\bar{x},\bar{y})$ &  $0.8223$ &   &  &  &  $1.1085$   &       &    \\
5th, 95th percentile of $\hat{\eta}_{m}(\bar{x},\bar{y})$ & $[-1.7791,3.9124]$ &&  & & $[0.4981,2.1654]$ &  & \\[0.3cm]
\multicolumn{3}{l}{$N=50,000$}\\[0.3cm]
mean $\hat{\eta}_{\pi}(\bar{x},\bar{y})$ &  $2.9020$ & $1.0492$   & $1.0492$ &  &  $0.8691$   &  $1.1375$     &  $1.1375$  \\
5th, 95th percentile of $\hat{\eta}_{\pi}(\bar{x},\bar{y})$ & $[-3.8346,8.4024]$ & $[-0.2543,3.2406]$ & $[-0.2543,3.2406]$ & & $[0.4641,2.4983]$ & $[0.4494,2.0255]$ & $[0.4494,2.0255]$\\
mean $\hat{\eta}_{\Lambda}(\bar{x},\bar{y})$ &  $2.2330$ & $-1.1672$   & $1.0916$ &  &  $1.2485$   &  $1.0592$     &  $1.1440$  \\
5th, 95th percentile of $\hat{\eta}_{\Lambda}(\bar{x},\bar{y})$ & $[-4.2416,5.5718]$ & $[-3.4766,5.5382]$ & $[-0.2671,3.3076]$ & & $[0.2255,2.6786]$ & $[0.3416,2.2478]$ & $[0.4696,2.0793]$\\
mean $\hat{\eta}_{m}(\bar{x},\bar{y})$ &  $1.4825$ &   &  &  &  $1.0984$   &       &    \\
5th, 95th percentile of $\hat{\eta}_{m}(\bar{x},\bar{y})$ & $[0.2295,4.2661]$ &&  & & $[0.6557,1.8400]$ &  & \\[0.3cm]
\hline
\hline
\multicolumn{7}{l}{\footnotesize Note: trimming I: boundary trimming, trimming II: boundary and denominator trimming}
\end{tabular}
\label{T:sim_res3}
\end{table}
\end{landscape}

\begin{landscape}
\begin{table}[!htbp]
\centering
\scriptsize
\caption{Simulation results for $\hat{\eta}_{\pi}$, $\hat{\eta}_{\Lambda}$ and $\hat{\eta}_m$ under the Clayton copula with $\tau = 0.8$. True $\eta = 1$.}
	\begin{tabular}{rccccccc}
		\hline\hline
& \multicolumn{3}{c}{$h=0.2$} & & \multicolumn{3}{c}{$h=0.3$}\\
& no trimming & trimming I & trimming II & & no trimming & trimming I & trimming II  \\
\hline
\multicolumn{3}{l}{$N=10,000$}\\[0.3cm]
mean $\hat{\eta}_{\pi}(\bar{x},\bar{y})$ &  $2.1767$ & $-0.0285$   & $-0.0284$ &  &  $-0.3005$   &  $0.1228$     &  $0.1228$  \\
5th, 95th percentile of $\hat{\eta}_{\pi}(\bar{x},\bar{y})$  & $[-11.1661,7.1969]$ & $[-2.4220,1.5839]$ & $[-2.4220,1.5839]$ & & $[-9.0571,7.1984]$ & $[-2.1346,1.9207]$ & $[-2.1346,1.9207]$\\
mean $\hat{\eta}_{\Lambda}(\bar{x},\bar{y})$  &  $0.9905$ & $-1.2620$   & $-0.0927$ &  &  $1.2997$   &  $0.2701$     &  $0.1424$  \\
5th, 95th percentile of $\hat{\eta}_{\Lambda}(\bar{x},\bar{y})$ & $[-3.8712,8.0429]$ & $[-19.5558,3.8210]$ & $[-2.4942,1.5550]$ & & $[-4.7490,5.5179]$ & $[-3.9731,6.6296]$ & $[-2.1905,2.0261]$\\
mean $\hat{\eta}_{m}(\bar{x},\bar{y})$ &  $-8.7675$ &   &  &  &  $-0.5581$   &       &    \\
5th, 95th percentile of $\hat{\eta}_{m}(\bar{x},\bar{y})$ & $[-8.9316,8.0082]$ &&  & & $[-3.2967,6.8237]$ &  & \\[0.3cm]
\multicolumn{3}{l}{$N=25,000$}\\[0.3cm]
mean $\hat{\eta}_{\pi}(\bar{x},\bar{y})$ &  $0.1857$ & $-0.1711$   & $-0.1711$ &  &  $4.2212$   &  $0.3599$     &  $0.3599$  \\
5th, 95th percentile of $\hat{\eta}_{\pi}(\bar{x},\bar{y})$ & $[-9.3946,11.8305]$ & $[-2.0212,2.0472]$ & $[-2.0212,2.0472]$ & & $[-5.9141,14.1959]$ & $[-1.5959,2.4209]$ & $[-1.5959,2.4209]$\\
mean $\hat{\eta}_{\Lambda}(\bar{x},\bar{y})$ &  $-5.9138$ & $-2.1169$   & $-0.1562$ &  &  $-0.1304$   &  $2.9763$     &  $0.4114$  \\
5th, 95th percentile of $\hat{\eta}_{\Lambda}(\bar{x},\bar{y})$ & $[-14.7657,6.9706]$ & $[-5.7411,2.7060]$ & $[-2.3037,2.0794]$ & & $[-9.7981,12.8433]$ & $[-1.9714,13.4558]$ & $[-1.2021,2.5000]$\\
mean $\hat{\eta}_{m}(\bar{x},\bar{y})$ &  $-0.2561$ &   &  &  &  $0.4879$   &       &    \\
5th, 95th percentile of $\hat{\eta}_{m}(\bar{x},\bar{y})$  & $[-8.3296,8.3212]$ &&  & & $[-4.5019,5.7117]$ &  & \\[0.3cm]
\multicolumn{3}{l}{$N=50,000$}\\[0.3cm]
mean $\hat{\eta}_{\pi}(\bar{x},\bar{y})$ &  $-0.8661$ & $0.2057$   & $.2057$ &  &  $0.4597$   &  $0.7847$     &  $0.7847$  \\
5th, 95th percentile of $\hat{\eta}_{\pi}(\bar{x},\bar{y})$ & $[-13.8462,3.4500]$ & $[-1.5599,2.5771]$ & $[-1.5599,2.5770]$ & & $[-6.8906,5.2867]$ & $[-0.6204,2.2352]$ & $[-0.6204,2.2352]$\\
mean $\hat{\eta}_{\Lambda}(\bar{x},\bar{y})$  &  $132.2612$ & $2.8237$   & $0.2112$ &  &  $-0.35672$   &  $0.67561$     &  $0.8162$  \\
5th, 95th percentile of $\hat{\eta}_{\Lambda}(\bar{x},\bar{y})$ & $[-16.3917,5.5396]$ & $[-5.3536,24.9165]$ & $[-1.6172,2.6039]$ & & $[-6.8886,4.3557]$ & $[-4.3075,4.8439]$ & $[-0.7286,2.3439]$\\
mean $\hat{\eta}_{m}(\bar{x},\bar{y})$ &  $-16.0181$ &   &  &  &  $0.8900$   &       &    \\
5th, 95th percentile of $\hat{\eta}_{m}(\bar{x},\bar{y})$  & $[-8.1384,5.4104]$ &&  & & $[-5.8509,6.4618]$ &  & \\[0.3cm]
\hline
\hline
\multicolumn{7}{l}{\footnotesize Note: trimming I: boundary trimming, trimming II: boundary and denominator trimming}
\end{tabular}
\label{T:sim_res4}
\end{table}
\end{landscape}

\pagebreak

\begin{figure}[!htbp]
	\centering
    (a) $h=0.2$ \hspace{7cm} \\
	\includegraphics[width=0.8\textwidth]{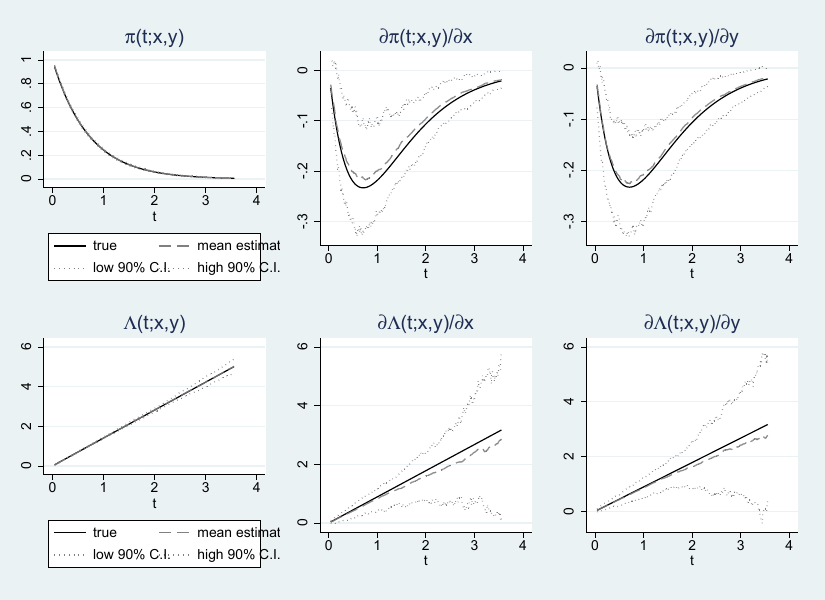} \vspace{0.2cm}\\
(b) $h=0.3$ \\
    \includegraphics[width=0.8\textwidth]{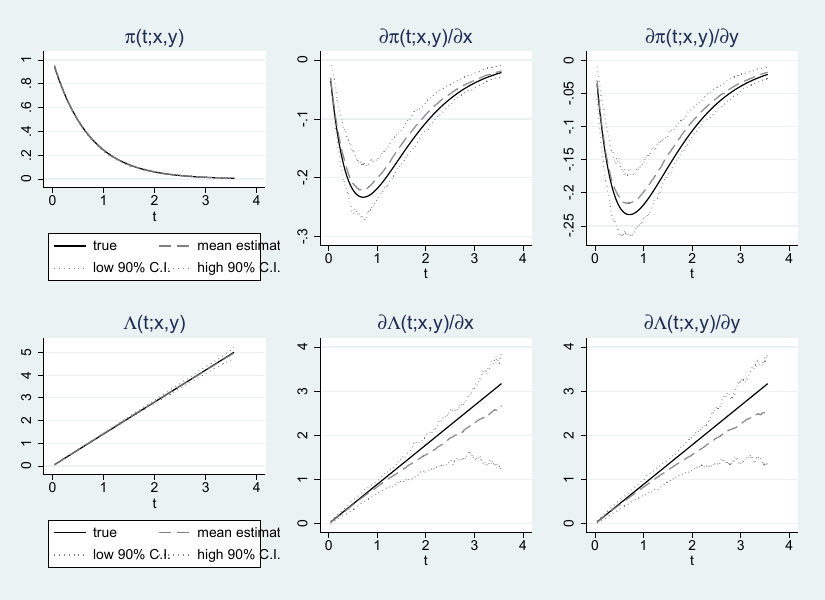}\\
	\caption{Simulation results for $N=50,000$, Gumbel with $\tau=0.1$ from 100 runs: true functions $\pi(t \mid \bar{x},\bar{y})$, $\Lambda(t \mid \bar{x},\bar{y})$, $\frac{\partial \pi(t\mid \bar{x},\bar{y})}{\partial z}$ for $z\in\{x,y\}$ and $\frac{\partial \Lambda(t\mid \bar{x},\bar{y})}{\partial z}$ for $z\in\{x,y\}$ in $t$ (solid black lines), mean of their nonparametric estimates (dashed gray lines) and 5'th and 95'th percentiles of the respective distributions (grey dots).}
	\label{F:sim1}
\end{figure}

\begin{figure}[!htbp]
    \centering
    (a) $\hat{\eta}_{\pi}(t,\bar{x},\bar{y})$ \hspace{4.5cm} (b) $\hat{\eta}_\Lambda(t,\bar{x},\bar{y})$ \\
    \includegraphics[scale=0.55]{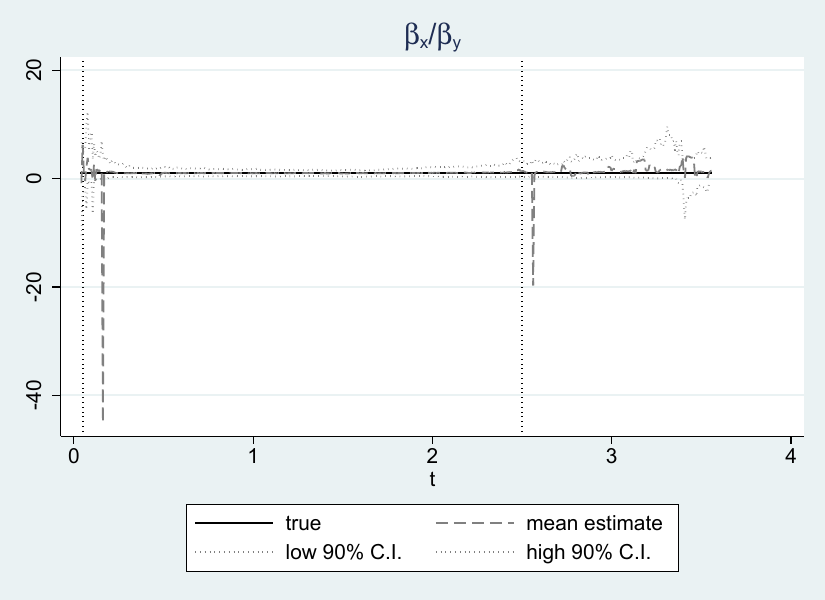} \hspace{0.2cm}
    \includegraphics[scale=0.55]{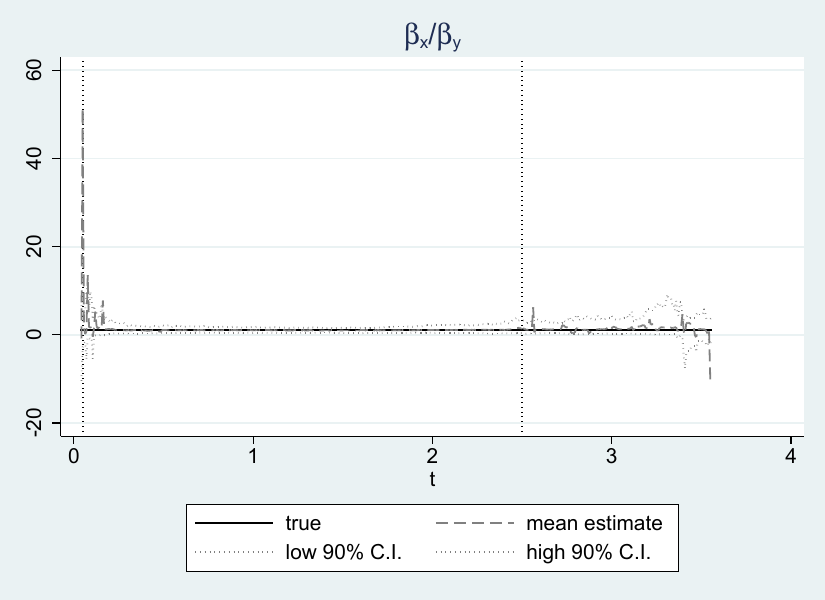}
	\caption{Simulation results for $N=50,000$ and $h=0.2$ under the Gumbel copula with $\tau=0.1$ from 100 runs: true $\eta_\pi (t,\bar{x},\bar{y}) = \beta_x/\beta_y = 1$ for all $t$ (solid black line), mean of $\hat{\eta}_\pi(t,\bar{x},\bar{y})$ and $\hat{\eta}_\Lambda(t,\bar{x},\bar{y})$ (dashed grey line) and 5th and 95th percentiles of the distribution of $\hat{\eta}_\pi(t,\bar{x},\bar{y})$ and $\hat{\eta}_\Lambda(t,\bar{x},\bar{y})$.}
	\label{F:sim2}
\end{figure}

\end{appendices}

\end{document}